\documentclass[11pt]{article}

\usepackage{times}
\usepackage{subfigure}
\usepackage{fullpage}
\usepackage{setspace}
\usepackage{color}
\usepackage{graphicx,amssymb,amsmath}
\usepackage{xspace}
\usepackage[linesnumbered,boxed,ruled,vlined]{algorithm2e} 

\usepackage{tikz}
\usepackage{tkz-graph}
\usetikzlibrary{matrix, arrows,shapes, positioning, automata}
\usepackage{geometry}
\usepackage{fullpage}
\usepackage[final]{hyperref} 
\usepackage{float}

\usepackage{appendix}
\usepackage{graphicx}
\usepackage{amsfonts}
\usepackage{amssymb}
\usepackage{amsmath}
\usepackage{amsthm}
\usepackage{url}
\usepackage{cite}
\usepackage{authblk}

\hypersetup{
	colorlinks=true,       
	linkcolor=blue,          
	citecolor=blue,        
	filecolor=magenta,      
	urlcolor=blue
}

\makeatletter
\newtheorem*{rep@theorem}{\rep@title}
\newcommand{\newreptheorem}[2]{%
\newenvironment{rep#1}[1]{%
 \def\rep@title{#2 \textbf{\ref{##1}}}%
 \begin{rep@theorem}}%
 {\end{rep@theorem}}}
\makeatother

\newreptheorem{theorem}{\textbf{Theorem}}

\newtheorem{lemma}{Lemma}
\newtheorem{thm}[lemma]{Theorem}
\newtheorem{defn}[lemma]{Definition}

\newtheorem{corollary}[lemma]{Corollary}

\newtheorem{remark}[lemma]{Remark}

\newcommand{\topic}[1]{\vspace{0.2cm}\noindent{\bf #1:}}

\newcommand{\e}{\epsilon}
\newcommand{\R}{\mathbb{R}}

\newcommand{\floor}[1]{\left\lfloor #1 \right\rfloor}
\newcommand{\ceil}[1]{\left\lceil #1\right\rceil}

\newcommand{\bigceil}[1]{\left\lceil #1\right\rceil}
\newcommand{\OPT}{\mathsf{OPT}}
\newcommand{\SOL}{\mathsf{SOL}}
\newcommand{\M}{\mathbb{M}}
\newcommand{\calM}{\mathcal{M}}
\newcommand{\calT}{\mathcal{T}}
\newcommand{\EMD}{\mathbf{EMD}}
\newcommand{\supp}{\mathbf{supp}}
\newcommand{\headtree}{\textsf{Tree-Sparsity-Head}}
\newcommand{\tailtree}{\textsf{Tree-Sparsity-Tail}}
\newcommand{\logapp}{ W }

\newcommand{\RS}{RS}
\newcommand{\layer}{part}
\newcommand{\layers}{parts}

\newcommand{\calC}{\mathcal{C}}
\newcommand{\bcalC}{\bar{\mathcal{C}}}

\newcommand{\emd}{\Delta}
\newcommand{\headvalue}{\Phi}

  \SetKwFunction{RSMinPlus}{RSMinPlus}
  \SetKwFunction{MinPlus}{MinPlus}
  \SetKwFunction{FindTree}{FindTree}

\definecolor{red}{rgb}{0.89, 0.0, 0.13}
\newcommand{\eat}[1]{}
\newcommand{\jian}[1]{{\bf\color{red} Jian: #1}}

\title{Improved Algorithms For Structured Sparse Recovery}
\author[1]{Lingxiao Huang}
\author[1]{Yifei Jin}
\author[1]{Jian Li\thanks{lijian83@mail.tsinghua.edu.cn}}
\author[2]{Haitao Wang\thanks{haitao.wang@usu.edu}}
\affil[1]{Tsinghua University, Beijing 100084, China}
\affil[2]{Utah State University, Logan, UT 84322, USA}

\begin{document}
\maketitle

\begin{abstract}
It is known that
	certain structures of the signal in addition to
	the standard notion of sparsity (called {\em
		structured sparsity}) can improve the sample
	complexity in several compressive sensing applications.
Recently, Hegde et al.~\cite{hegde2015approximation} proposed a framework, called \emph{approximation-tolerant model-based
compressive sensing}, for recovering signals with
structured sparsity.
Their framework requires two oracles, the \emph{head-} and the \emph{tail-approximation projection oracles}. The two
oracles should return approximate solutions in the model which is closest to the query signal. In this paper, we consider two structured
sparsity models and obtain improved projection algorithms.
The first one is the \emph{tree sparsity model}, which captures the
support structure in the wavelet decomposition of piecewise-smooth signals and images.
We propose a linear  time $(1-\e)$-approximation algorithm for head-approximation projection and a linear time $(1+\e)$-approximation algorithm for tail-approximation projection.
The best previous result is an $\tilde{O}(n\log n)$ time
{\em bicriterion} head-approximation (tail-approximation) algorithm
(meaning that their algorithm may return a solution of sparsity larger than $k$)
by Hegde et al~\cite{hegde2014nearly}.
Our result provides an affirmative answer to the open problem mentioned in the survey of Hegde and
Indyk~\cite{hegde2015fast}.
As a corollary, we can recover a constant approximate $k$-sparse signal.
The other is the \emph{Constrained Earth Mover Distance (CEMD) model}, which is useful to model the situation where
the positions of the nonzero coefficients of a signal do not change
significantly as a function of spatial (or temporal) locations. We obtain the first single criterion constant factor
approximation algorithm for the
head-approximation projection~\cite{hegde2015approximation}. The previous best known algorithm is a bicriterion approximation.
Using this result, we can get a faster constant approximation algorithm with fewer measurements
for the recovery problem in CEMD model.

\end{abstract}


\section{Introduction}
\label{sec:intro}

\eat{
Compressive Sensing states that the sparsity of a signal can be exploited to recover it from far
fewer samples than required by the Shannon-Nyquist sampling
theorem~\cite{candes2006robust,donoho2006compressed}. Specifically, \emph{linear measurements} for
the signals is well studied over the last decades~\cite{foucart2013mathematical}. In this setting,
the representation of a signal $x$ is given by $y = Ax$, where $A$ is a $m\times n$ measurement
matrix and $y \in \R^{m}$.  The signal in real world contains noise.
}

We consider the \emph{robust sparse recovery}, an important
problem in compressive sensing.
The goal of robust sparse recovery is to recover a signal from a small number of linear measurements. Specifically, we
call a vector $x\in \R^n$ $k$-sparse if it has at most $k$ non-zero entries.  The support of $x$, denoted by $\supp(x) \subseteq [n]$,
contains the indices corresponding to the nonzero entries in $x$. Given the measurement vector $y=Ax
+ e$, where $A$ is a measurement matrix, $x$ is a $k$-sparse signal  and $e$ is a
noise vector, the goal is to find a signal estimate $\hat{x}$ such that $ \Vert x - \hat{x}\Vert < C\| e\|$ for some constant approximation factor $C>0$.

\eat{
Prior results show that we can recover a $k$-sparse signal
from $m=O(k\log (n/k))$ measurements and it is tight~\cite{do2010lower}.
}

\topic{Structured Sparsity Model} It is well known that in general we need the number of measurements to be $\Omega(k\log (n/k))$ for robust sparse recovery, see \cite{do2010lower,foucart2010gelfand}. In practice, the support of $x$ usually has some
structured constraints, such as tree sparsity and block sparsity, which can help reduce the bound of the number of measurements.

\begin{defn}[Structured Sparsity Model \cite{baraniuk2010model}]
  \label{def:model}
  Let $\mathbb{M}$ be a
  family of supports, i.e., $\mathbb{M} = \{\Omega_1, \Omega_2, $ $\ldots
  ,\Omega_L\}$ where each $\Omega_i \subseteq [n]$. Then the corresponding structured sparsity model
  $\mathcal{M}$ is the set of vectors supported on one of the $\Omega_i$:
  \begin{equation*}
    \mathcal{M} = \{x \in \mathbb{R}^d \mid \mathbf{supp}(x) \subseteq \Omega \text{ for
      some } \Omega  \in \mathbb{M}\}.
  \end{equation*}

\end{defn}

To recover such a structured signal $x$ is called structured sparse recovery. Baraniuk et
al.~\cite{baraniuk2010model} provided a general framework called model-based
compressive sensing. Their
framework depends on a \emph{model projection oracle} which is defined as follows.
\begin{defn}[Model Projection \cite{baraniuk2010model}]
  \label{def:modelproj}
  Let $\M$ be a structured sparsity model. A model projection oracle for $\M$ is an
  algorithm $P(x): \mathbb{R}^n \rightarrow \M$  such that $\Omega^* = P(x) $ and
  \begin{equation*}
    \| x - x_{\Omega^*} \|_p = \min_{\Omega \in \M } \| x - x_{\Omega}\|_p
  \end{equation*}
  where $x_{\Omega}\in \R^n $ is the same as $ x$ on the support $\Omega$, and is zero otherwise.
\end{defn}

Unfortunately, the best known algorithms for many exact model projection oracles are too slow to be used
in practice. Some of the exact model projection oracles are even NP-hard.
Recently, Hegde et al.~\cite{hegde2015approximation} provided a
principled method AM-IHT for recovering
structured sparse signals. Their framework only requires two approximation oracles called the
head- and tail-approximation projection oracles, defined as follows. Let $\Omega^*\in \M$ be the optimal support of the model projection oracle as defined in Definition~\ref{def:modelproj}.

\begin{defn}[Head-Approximation Projection]
  \label{def:head}
  Let $\M$ be a structured sparsity model. A head-approximation oracle for $\M$ is
  an algorithm $H(x):\R^n \rightarrow \M $ such that $H(x) = \Omega$ and
  \begin{equation*}
    \| x_{\Omega} \|_p \geq c_{H} \cdot \| x_{\Omega^*} \|_p
  \end{equation*}
  where $c_H \in (0,1]$ is a  fixed constant.
\end{defn}

\begin{defn}[Tail-Approximation Projection]
  \label{def:tail}
  Let $\M$ be a structured sparsity model. A tail-approximation oracle for $\M$ is
  an algorithm $T(x): \R^n \rightarrow \M$  such that $T(x) = \Omega $ and
  \begin{equation*}
    \|x - x_{\Omega} \|_p \leq c_T \cdot \| x - x_{\Omega^*} \|_p,
  \end{equation*}
  where $ c_T \in [1,\infty) $ is a fixed constant.
\end{defn}

\subsection{Tree Sparsity Model and CEMD Model}
\label{sec:model}

\topic{Tree Sparsity Model}
The tree sparsity model can be used for capturing the support structure of the wavelet decomposition of piecewise-smooth signals and images~\cite{bohanec1994trading, chen2012compressive, hegde2014fast}.
In this model, the coefficients of the signal $x$
are arranged as the nodes of a complete $b$-ary tree $T$ rooted at node $N$, and any feasible solution is a subtree which includes the root of $T$ and is of size $k$.

\begin{defn}[Tree Sparsity Model]
  \label{defn:treesparse}
  Let $T$ be a complete $b$-ary tree with $n$ nodes rooted at node $N$. $\mathbb{T}_k (T) = \{
  \Omega_1, \Omega_2, \ldots,\Omega_L \}$ is the family of supports where each  $\Omega_i$ is a
  subtree of $T$ rooted at  $N$ with  the number of nodes  no more than $k$. We use $\mathbb{T}_k$ instead of $\mathbb{T}_k (T)$ for short. The
  tree-structured sparsity model $\calT_k$ is the set of signals supported on some $\Omega\in \mathbb{T}_k$:
  \begin{equation*}
    \calT_k = \{ x \in \R^{d} \mid \supp(x) \subseteq \Omega \text{ for some } \Omega \in \mathbb{T}_k  \}
  \end{equation*}
\end{defn}

For the tree sparsity model,
the \emph{head-} and
\emph{tail-approximation projection} problems reduces
to the following simple-to-state combinatorial problems:
for the head-approximation,
we want to find a subtree of size $k$
rooted at $N_1$ such that the total weight of the subtree is maximized;
for the tail-approximation,
we want the total weight of the complement of the subtree is minimized.
For convenience, we abbreviate
them as \headtree\ and \tailtree\ respectively. If our solution is a subtree of size at most $k$, we call it a single-criterion solution. Otherwise if our solution is a subtree of size larger than $k$, we call it a bicriterion solution.

\topic{Constrained EMD Model}
The CEMD model, introduced by Schmidt et al~\cite{schmidt2013constrained}, is particularly useful in 2D image compression and denoising~\cite{vaswani2010modified,duarte2005distributed}.
We first introduce the definition of Earth Mover's Distance
(EMD), also known as the Wasserstein metric or Mallows distance~\cite{levina2001earth}.

\begin{defn} [EMD]
  \label{def:emd}
  The EMD of two finite sets $A,B\subset \mathbb{N}$ with $|A|=|B|$ is defined as
  $$
  \EMD(A,B)=\min_{\pi:A\rightarrow B} \sum_{a\in A} |a-\pi(a)|,
  $$
  where $\pi$ ranges over all one-to-one mappings from $A$ to $B$.
\end{defn}

In the CEMD model,  the signal $x\in \R^{n}$ can be interpreted as a matrix $X \in \R^{h\times w}$
with $n=hw$. By this interpretation, the support of $x\in \R^{h\times w}$, denoted by $\supp(x) \subseteq [h]\times [w]$,
contains the indices $(i,j)$ ($i\in [h],j\in [w]$) corresponding to the nonzero entries in $x$.

\eat{
  Note that $\EMD(A,B)$ is the cost of a minimum matching between $A$ and $B$. We consider the case
  where the sets $A$ and $B$ are the \emph{supports} of two exactly $s$-sparse signals. Thus, we have
  that $|A|=|B|=s$. In this case, the EMD measures how far the supported indices move. This notation
  can be generalized to an \emph{ensemble} of sparse signals~\cite{hegde2015approximation}.
}

\begin{defn} [Support-EMD]
  \label{def:semd}
  Consider an $ h \times w$ matrix $X$.
  Let $\Omega\subseteq [h]\times [w]$ be the support of a matrix $X$.  Denote $\Omega_i$ to be the
  support of the column $i$ of $X$. Suppose $ | \Omega_i |= s$ for $i \in [w]$.  Then the EMD of the
  support $\Omega$ (or the support-EMD of $X$) is defined as
  $$
  \EMD[\Omega]=\sum_{i=1}^{w-1} \EMD(\Omega_i, \Omega_{i+1}).
  $$
\end{defn}

\eat{
By the definition of Support-EMD, we naturally consider the following structured sparsity
model. Assume that the signal $x\in \R^{n}$ can be interpreted as a matrix $X \in \R^{h\times w}$
with $n=hw$. The total sparsity of $x$ is $k$ and the sparsity of each column is $s=k/w$.  The support-EMD of $x$ is at most $B$.
\jian{last half sentence not good}
}

Naturally, we have the following structured sparsity model which contains two constraints: 1) each
column is $s(=k/w)$-sparse, 2) the support-EMD is at most $B$.

\begin{defn} [Constrained EMD Model~\cite{schmidt2013constrained}]
  \label{def:cemd}
  Let $\M_{k,B}$ be the family of supports $\{\Omega\subseteq [h]\times [w] \mid \EMD(\Omega)\leq B, \text{ and }
  |\Omega_i|=k/w, \text{ for } i \in[w] \}$. The Constrained EMD (CEMD) model $\calM_{k,B}$ is the set of signals supported on some $\Omega\in \M_{k,B}$:
  \begin{equation*}
    \calM_{k,B} = \{ x \in \R^{d} \mid \supp(x) \subseteq \Omega \text{ for some } \Omega \in \M_{k,B}  \}
  \end{equation*}
\end{defn}

Consider the CEMD model projection problem. If our solution belongs to $\calM_{k,B}$, we call it a single-criterion solution. Otherwise if our solution does not belong to $\calM_{k,B}$, i.e., there exists a column of sparsity larger than $s(=k/w)$ or the support-EMD is larger than $B$, we call it a bicriterion solution.

\subsection{Our Contributions and Techniques}
For both the tree sparsity and CEMD models, we consider the corresponding model-projection problems. We obtain improved approximation algorithms,
which have faster running time, and return single-criterion solutions (rather than bicriterion solutions).
Consequently, combining with the AM-IHT framework \cite{hegde2015approximation},
our results implies better structured sparse recovery algorithms,
in terms of the number of measurements, the sparsity
of the solution, and the running time. We summarize our contributions and main techniques in the following.

\eat{
\begin{defn}[Approximation Criterion]
  We say an approximation model projection algorithm (or sparse recovery algorithm) for $k$-sparse signal  is
  bicriterion approximation if it returns a solution whose size of support  is twice over the  input
  parameter $k$. Moreover, we say an algorithm is  single-criterion if it returns a solution whose size of support is no more than $k$
  strictly.
\end{defn}
\jian{remove this defintion. don't make trivial def}

\eat{
  Note that in model-based compressive sensing, an efficient single-criterion algorithm for the model
  projection problem can reduce the total number of measurements, and decrease the running time of
  structured sparse recovery. Thus, our results are useful for structured sparse recovery and
  model-based compressive sensing.
  \jian{
	a very unprofessional sentence.
	what is the difference between the two things???}
}
}

\topic{Tree Sparsity Model}
\eat{
 We first consider the tree sparsity model.
Our main technical results are summarized as follows and the details are in Section~\ref{sec:tree}.
}
Cartis et
al.~\cite{cartis2013exact} gave an exact tree-sparsity projection algorithm with running time $O(nk)$. For the approximation version, Hegde et al.~\cite{hegde2014fast,hegde2014nearly} proposed bicriterion approximation schemes for both head- and tail-approximation tree-sparsity projection problems with running time $\tilde{O}(n\log n)$. Both algorithms achieve constant approximation ratio and output a tree of size at most $2k$.
In this paper, we provide the first linear time algorithms for
both head- and tail-approximation tree-sparsity projection problems and remove the bicriterion
relaxation.
This provides an affirmative answer to the open problem
in Hegde and Indyk~\cite{hegde2015fast}, which asks
whether there is a nearly-linear time single-criterion approximation algorithm for tree sparsity.

\topic{Main Techniques for Tree Sparsity Model} The bottleneck of previous algorithms is computing exact $(\min, +)$-convolutions. Our main technique is to improve the running time of $(\min, +)$-convolutions. In Section \ref{sec:conv}, we introduce an approach of computing an approximate $(\min, +)$-convolution, called $(\alpha,\beta)$-\RS\ $(\min,+)$-convolution.
Instead of maintaining the whole $(\min, +)$-convolution array, we only compute a sparse sequence to approximately represent the whole array. Taking \tailtree\ as an example, we only need to maintain $\tilde{O}(\log n)$ elements in a single node, instead of $k$ elements for the exact $(\min, +)$-convolution. For the computation time, we show that the running time of computing each convolution element can be reduced to $\tilde{O}(1)$, instead of $O(k)$ for the exact $(\min, +)$-convolution. Thus, we only cost $\tilde{O}(\log n)$ to compute our approximate $(\min, +)$-convolution. For \headtree, we apply a similar approximate $(\max, +)$-convolution technique, called $(\alpha,\beta)$-\RS\ $(\max,+)$-convolution. Our approximate convolution technique may have independent interest.

In Section \ref{sec:tree}, we combine the approximate $(\min, +)$-convolution technique and other approaches such as weight discretization, pruning and the lookup table method. Our results can be summarized by the following theorem.

\eat{Recall that the  \emph{fully polynomial-time
    approximation scheme or FPTAS} for a optimization problem is an algorithm which  takes an instance
  of the problem and a parameter $\e$, in polynomial time in both problem size $n$ and approximation
  ratio $1/\e$, to produce a $(1+\e)$-approximation solution.
  \jian{you already said linear time (1+e) approximation. no need to say FPTAS then.}
}
\eat{
Then, we state  our result formally. W.l.o.g., we assume that $k\geq \log n$.
\footnote{Otherwise, we could reduce the problem to this case by removing the nodes of
  the tree with depth larger than k.}
}

\begin{thm}[Linear time head- and tail-approximation tree-sparsity projection]
  \label{thm:linear}
  There are linear   time algorithms  for both head- and tail-approximation tree-sparsity projection problems.
  Specifically, for any constant
  $\e_1 \in (0,1)$,
  there is an $O(\e_1^{-1} n)$ time approximation algorithm  that returns a support $\hat{\Omega} \in
  \mathbb{T}_k$ satisfying
  \begin{equation*}
    \| x_{\hat{\Omega}}\|_p \geq (1- \e_1) \max_{\Omega \in \mathbb{T}_k} \| x_{\Omega}\|_p.
  \end{equation*}
  For any constant $\e_2 \in (0,\infty)$, there is an $O(n+ \e_2^{-2}n /\log n)$
  time approximation  algorithm that  returns a support $\hat{\Omega} \in \mathbb{T}_k$ satisfying
  \begin{equation*}
    \| x - x_{\hat{\Omega}} \|_{p} \leq (1+\e_2) \min_{\Omega \in \mathbb{T}_k} \| x
    - x_{\Omega} \|_p ,
  \end{equation*}
  if $k \leq n^{1-\delta}$ ($\delta \in(0,1)$ is any fixed constant), and there is an $O( \e_2^{-1}n(\log\log\log n)^2 )$ time algorithm for general $k$.
\end{thm}

Then combining with prior  results~\cite{baraniuk2010model,hegde2014nearly,hegde2015approximation}, we  provide
a more efficient robust sparse recovery algorithm in tree sparsity model as follows. The best prior
result can recover an approximate signal $\hat{x} \in \calT_{ck}$ for some constant
$c >1$~\cite{hegde2014nearly}
(i.e.,
the sparsity of their solution is $ck$). In this paper, we improve the constant $c$ to $1$.

\begin{corollary}
  \label{thm:tree}
   Assume that $k \leq n^{1-\delta}$ ($\delta \in(0,1)$ is any fixed constant). Let $A\in \R^{m\times n}$ be a measurement matrix. Let $x \in \calT_{k}$ be an arbitrary signal in
  the tree sparsity model with dimension $n$, and let $y=Ax+e\in \R^m$ be a noisy measurement
  vector. Here $e\in \R^m$ is a noise vector. Then there exists an algorithm to recover a signal
  approximation $\hat{x}\in \calT_{k}$ satisfying $\|x-\hat{x}\|\leq C\|e\|_2$ for some constant $C$ from
  $m=O(k)$ measurements. Moreover, the algorithm runs in $O( (n\log n+k^2\log n \log^2(k\log n))\log
  \frac{\|x\|_2}{\|e\|_2})$ time.
\end{corollary}

\eat{

\topic{Outline of Our Techniques}
We achieve the above results by several techniques, such as dynamic programming,
\RS\ $(\min,+)$-convolution and lookup table method. The exact tree projection algorithm
in~\cite{cartis2013exact} is based on $(\min,+)$-convolution.
However, it takes $O(n^2)$ time for each node.
In the paper, we deal with this problem by maintaining  a new data
structure called \RS\ $(\min,+)$-convolution. We prove that the \RS\ $(\min,+)$-convolution is a good
approximation for $(\min, +)$-convolution.
In the beginning, we discretize the weight of nodes such that the range of
the new weight is  bounded. We prove that the optimal solution of the new weight is a good
approximation for the original problem. Then, we partition the tree into different \layers. In the lower \layer, we use a classical technique called
lookup table method. For each node in the upper \layer, we compute  \RS\
$(\min,+)$-convolution.  Benefiting from the new weight, we can compute all of them in linear time.
\jian{we can probably remove the paragraph as there is not much new here...unfortunately.}
}

\topic{CEMD Model} In Section \ref{sec:emd}, we consider the CEMD model $\calM_{k,B}$ and propose the first single-criterion
constant factor approximation algorithm for
the head-approximation oracle.

\eat{
At first, we use a flow technique in
\cite{hegde2015approximation}. Then we obtain a single-criterion solution via some nice structured
properties. The details can be found in Section~\ref{sec:emd}. We summarize our results as follows.
}

\begin{thm}
  \label{thm:emdcon}
  Consider the CEMD model $\calM_{k, B}$ with $s=k/w$ sparse for each column and support-EMD $B$.
  Let $\delta \in (0,1/4)$, $x_{\min}=\min_{|X_{i,j}|>0} |X_{i,j}|^{p}$, and $x_{\max}=\max
  |X_{i,j}|^{p}$. Let $c=1/4 -\delta$. There exists an algorithm running in
  $O(shn \log\frac{n}{\delta}+\log \frac{x_{\max}}{x_{\min}})$ time, which returns a single-criterion
  $c^{1/p}$ approximation for the  head-approximation projection problem. 
\end{thm}

Combining with AM-IHT framework~\cite{hegde2015approximation},
	we obtain the following corollary which improves the prior
result~\cite{hegde2015approximation} in two aspects:
1) We decrease
the total number of measurements from $m=O(k\log (\frac{B}{k}\log \frac{k}{w}))$ to $m=O(k\log
(B/k))$. 2) We decrease the running time of the robust sparse recovery from $O(n\log
  \frac{\|x\|_2}{\|e\|_2}(k\log n+\frac{kh}{w}(B+\log n+\log \frac{x_{\max}}{x_{\min}})))$ to $O(n\log
  \frac{\|x\|_2}{\|e\|_2}(k\log n+\frac{kh}{w}(\log n+\log \frac{x_{\max}}{x_{\min}})))$.

\begin{corollary}
  \label{thm:CEMD}
  Let $A\in \R^{m\times n}$ be a measurement matrix. Let $x\in \calM_{k,B}$ be an arbitrary signal in
  the CEMD model with dimension $n=wh$, and let $y=Ax+e\in \R^m$ be a noisy measurement vector. Here
  $e\in \R^m$ is a noise vector. Then there exists an algorithm to recover a signal approximation
  $\hat{x}\in \calM_{k, 2B}$ satisfying $\|x-\hat{x}\|\leq C\|e\|_2$ for some constant $C$ from
  $m=O(k\log (B/k))$ measurements. Moreover, the algorithm runs in $O(n\log
  \frac{\|x\|_2}{\|e\|_2}(k\log n+\frac{kh}{w}(\log n+\log \frac{x_{\max}}{x_{\min}})))$ time, where
  $x_{\max}=\max |x_i|$ and $x_{\min}=\min_{|x_i|>0} |x_i|$.
\end{corollary}

\subsection{Related Work}

In the tree sparsity model, there is an algorithm with running time  $O(nk\log n)$ for the exact model projection  by dynamic programming.
By using a more careful analysis, Cartis et
al.~\cite{cartis2013exact} improved exact model projection to $O(nk)$.
Actually, their dynamic program was
based on computing $(\min, +)$-convolutions. The naive algorithm for computing the $(\min, +)$-convolution of arbitrary two
length-$n$ arrays requires $O(n^2)$ time. Williams
proposed an improvement algorithm for computing $(\min, +)$-convolutions, and reduced the running time to
$O(n^2/2^{\Omega(\sqrt{\log n})})$ \cite{williams2014faster}.

For the approximation
projection problem, several heuristic algorithms had been proposed, such as CSSA~\cite{baraniuk1993signal},
CPRSS~\cite{donoho1997cart}, optimal-pruning\cite{bohanec1994trading}.
Hegde et al.~\cite{hegde2014fast,hegde2014nearly} improved the running time of tree sparse recovery to
$\tilde{O}(n\log n)$.

Schmidt et al.~\cite{schmidt2013constrained} introduced the Constrained Earth Mover's Distance
(CEMD) model. Hegde, Indyk and Schmidt~\cite{hegde2015approximation} proposed \emph{bicriterion} approximation algorithms
for both head- and tail-approximation projections. How to find a single-criterion approximation algorithm is
an open problem mentioned in the survey~\cite{hegde2015fast}.

Other structured sparsity models also have been studied by researchers. Huang et al. \cite{huang2011learning} first considered the graph sparsity model, and provided a head-approximation algorithm with a time complexity of $O(n^c)$, where $c>1$ is a trade-off constant between time and sample complexity. For the tail-approximation projection problem, Hedge et al. \cite{hegde2015nearly} proposed a nearly-linear time bicriterion algorithm with the tail-approximation guarantee by modifying the GW scheme \cite{goemans1995general}. Hedge et al. \cite{hegde2009compressive} also studied the $\triangle$-separated model and provided an exact model projection algorithm.

Very recently during SODA17 conference, we knew that, in parallel to our work, Backurs et al. \cite{backurs2017better} also provided 
single criteria algorithms for the tree sparsity problem. Their algorithms can handle for more general trees and run in time $n(\log n)^{O(1)}$. 
Our algorithms only work for b-ary trees, but our running times are much better.  



\section{Approximate $(\min, +)$-Convolution}
\label{sec:conv}

In this section, we introduce an approach of computing an approximate $(\min, +)$-convolution, which is useful for the tree sparsity model.

\subsection{$(\alpha,\beta)$-\RS\ $(\min,+)$-Convolution}
\label{sec:appconv}

We first introduce a concept called ($\min, +$)-convolution.

\begin{defn}[$(\min, +)$-convolution (see e.g.~\cite{bremner2014necklaces,chan2015clustered})]
  \label{def:min+}
  Given two  arrays $A = (a[0], a[1], a[2], $ $\ldots, $$ a[m_1])$ and $B = (b[0], b[1], b[2], \ldots,
  b[m_2])$, their  $(\min,
  +)$-convolution  is the array $S = (s[0], s[1], s[2], \ldots , $ $s[m_1+m_2])$ where $s[t] =
  \min_{i=0}^{t} \{a[i] + b[t-i]\} , t \in  [0, m_1+m_2]$.
\end{defn}

\eat{A similar problem is called $(\max, +)$-convolution. The only difference is that $s_k = \max_i \{a_i + b_{k-i}\}$ in $(\max,
  +)$-convolution. \footnote{Note that these two convolutions are equivalent. Since we can use $A' = (-a_1, -a_2,
    \ldots, -a_m)$ and $B' = (-b_1, -b_2, \ldots, -b_n)$ as the input of $(\max, +)$-convolution, and
    the output $s'_k$ is equal to $-s_k$ as the output of the $(\min, +)$-convolution.}}

We sketch how to use $(\min, +)$-convolutions in the tree sparsity model. Recall that $T_{ij}$ is
the subtree rooted at
$N_{ij}$. We maintain an array $S_{ij}=(s[0], s[1],\ldots, s[|T_{ij}|])$ for each node $N_{ij}$. The element
$s[l]$ represents the optimal tail value for \tailtree\ on $T_{ij}$, i.e., $s[l]=\min_{\Omega\in
  \mathbb{T}_{|T_{ij}|-l} (T_{ij})} $ $ \sum_{N_{i'j'} \in T_{ij} \setminus \Omega} x_{i'j'}$ where
$\mathbb{T}_{k} (T_{ij})$  is the tree sparsity model defined at the tree $T_{ij}$ (see
Definition~\ref{defn:treesparse}). In fact, the array $S_{ij}$
can be achieved through computing the $(\min, +)$-convolution from the arrays of its two
children.
  \footnote{Note that the value of $s[|T_{ij}|]$ can not be directly obtained by the $(\min,
    +)$-convolution of the arrays of its two
children. In fact, $s[|T_{ij}|]=\sum_{N_{i'j'} \in T_{ij}}x_{i'j'}$.}
Finally, we output the element $s[n-k]$ in the array of the root node $N_{\log(n+1), 1}$, which is the optimal
tail value of \tailtree\ on $T$.

However, the running time for computing the exact $(\min, +)$-convolution is too long. Instead, we
compute an approximate $(\min, +)$-convolution for each node. We first introduce some
concepts.

\eat{
  Note that
  we can consider the two arrays $A$ and $B$ as function curves with support $[n]$ and images $A$ and
  $B$ respectively. See Figure~\ref{fig:min+}.  Actually, we need not maintain the curve exactly.
  Instead, maintaining some of the step points of the curve is enough since tolerating small
  inexactness.
}

\begin{defn}[$\alpha$-\RS]
Given a sequence $\hat{A} =
  (\hat{a}[i_1], \hat{a}[i_2], \ldots , \hat{a}[i_m])$ and a fixed constant $\alpha  \in [0,\infty)$.
  Each element $\hat{a}[i_v]\in \hat{A} $ is a real number with an associated
  index $i_v $. If for any $ v \in [1, m-1]$,  $i_{v+1} > i_v, \hat{a}[i_{v+1}] \geq
  (1+\alpha)\hat{a}[ i_v]\geq 0$,  we call the sequence $\hat{A}$ an $\alpha$-representative
  sequence ($\alpha$-\RS).
\end{defn}

\begin{defn}[Completion of $\alpha$-\RS]
  \label{def:extended}
  Consider an $\alpha$-\RS\ $\hat{A} =(\hat{a}[i_1], \hat{a}[i_2] \ldots,  \hat{a}[i_m])$. Define its
  completion by an array $ A' =(a'[0], a'[1], \ldots, a'[i_m])  $ satisfying that: 1) If $0\leq t \leq i_1, a'[t] = a'[i_1]$; 2)  If $i_v+1\leq t \leq i_{v+1} \ (1\leq v\leq m-1), a'[t] = \hat{a}[i_{v+1}]$.
\end{defn}

By the following definition, we show how to use an $\alpha$-\RS\ to approximately represent an array.

\begin{defn}[Sequence  Approximation]
\label{def:seqapp}
  Given two non-decreasing arrays $A' = (a'[0], a'[1], \ldots, $ $ a'[n])$ and $A = (a[0], a[1], \ldots,
  a[n])$, we say $A'$ is an $\alpha$-approximation of  $A$  if for any $i$, $a[i] \leq a'[i]  \leq
  (1+\alpha)a[i] $. We say an $\alpha$-\RS\ $\hat{A}$ approximates an array $A$ if its completion $A'$ is an $\alpha$-approximation of $A$.
\end{defn}

A special case is that we say $A'$ is a $0$-approximation of $A$ if $A'=A$. Figure~\ref{fig:min+} illustrates these concepts. Now, we are ready to introduce the formal definition of the approximate
$(\min,+)$-convolution, called $(\alpha, \beta)$-\RS\ $(\min, +)$-convolution.

\begin{figure}[t]
  \centering
  \begin{tikzpicture}
    \draw[<->] (0,4.5) node[left]{value} -- (0,0) -- (9,0) node[below]{$k$} ;
    \draw[dashed] (0,0) -- (1,0) -- (1,0.3) -- (2,0.3) -- (2,0.7) --  (3,0.7) -- (3, 1) -- (4,1) --
    (4, 2.2) -- (5,2.2) -- (5, 2.7) -- (6, 2.7) -- (6, 3.2) --(7,3.2) --(8,3.2) --(8,4)  -- (8,0) ;
    \draw[thick] (0,0) --(1,0) -- (1,0.3 )-- (2,0.3)--(2, 1) -- (4 ,1) -- (4 , 4) -- (8,4);

    \draw[fill] (0,0) circle (0.1)  ;        \node[gray, above ] at (0,0) {$a'(0)$} ;
    \draw[fill] (1,0.3) circle (0.1)  ;     \node[gray, above ] at (1,0.3) {$a'(1)$} ;
    \draw[fill] (3,1) circle (0.1)  ;        \node[gray, above ] at (3,1) {$a'(3)$} ;
    \draw[fill] (8,4) circle (0.1) ;     \node[gray, above  ] at (8,4) {$a'(8)$} ;

    \draw[fill] (0,0) circle (0.1)  ;        \node[red, above ] at (0,0.4) {$\hat{a}(0)$} ;
    \draw[fill] (1,0.3) circle (0.1)  ;     \node[red, above ] at (1,0.7) {$\hat{a}(1)$} ;
    \draw[fill] (3,1) circle (0.1)  ;        \node[red, above ] at (3,1.4) {$\hat{a}(3)$} ;
    \draw[fill] (8,4) circle (0.1) ;     \node[red, above  ] at (8,4.4) {$\hat{a}(8)$} ;

    \draw[fill] (2, 1) circle (0.05) ;    \node[ gray, above] at (2,1) {$a'(2)$} ;
    \draw[fill] (4,4) circle (0.05) ;        \node[gray, above ] at (4,4) {$a'(4)$} ;
    \draw[fill] (5,4) circle (0.05) ;        \node[gray, above ] at (5,4) {$a'(5)$} ;
    \draw[fill] (6,4) circle (0.05) ;      \node[gray, above  ] at (6, 4) {$a'(6)$} ;
    \draw[fill] (7,4) circle (0.05) ;     \node[gray, above  ] at (7,4) {$a'(7)$} ;

    \draw[fill] (2, 0.7) circle (0.05) ;
    \draw[fill] (4, 2.2) circle (0.05) ;
    \draw[fill] (5, 2.7) circle (0.05) ;
    \draw[fill] (6, 3.2) circle (0.05) ;
    \draw[fill] (7, 3.2) circle (0.05) ;
    \draw[fill] (8, 4) circle (0.05) ;

    \node[blue, below ] at (0,0) {$a(0)$} ;
    \node[blue, below ] at (1 ,0.3) {$a(1)$} ;
    \node[blue, below ] at (2,0.7) {$a(2)$} ;
    \node[blue, below ] at (3,1) {$a(3)$} ;
    \node[blue, below ] at (4,2.2) {$a(4)$} ;
    \node[blue, below ] at (5, 2.7) {$a(5)$} ;
    \node[blue, below ] at (6,3.2) {$a(6)$} ;
    \node[blue, below ] at (7,3.2) {$a(7)$} ;
    \node[blue, below ] at (8,4) {$a(8)$} ;

    \node [below] at (8,0) {$n = 8$} ;

  \end{tikzpicture}

  \caption{ The figure illustrates the concepts $\alpha$-\RS\ and its completion. Here, $\hat{A} = (\hat{a}[0], \hat{a}[1], \hat{a}[3],\hat{a}[8])$ is an  $\alpha$-\RS.  The array $(a'[0]=\hat{a}[0], a'[1],\ldots, a'[8])$  is the completion of $\hat{A}$. By this figure, we can see that the $\alpha$-\RS\ $\hat{A}$ approximates the array $A =  (a[0], a[1], \ldots, a[8])$. }

  \label{fig:min+}
\end{figure}
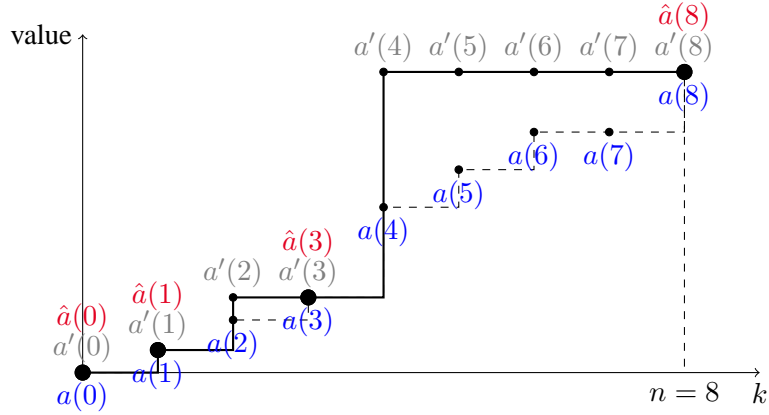

\begin{defn}[$(\alpha,\beta)$-\RS\ $(\min,+)$-convolution]
  \label{def:sparsemin} Given two $\alpha$-\RS s   $\hat{A} $  and $ \hat{B} $, suppose $A'$ and $B'$ are
  their completions respectively.  Suppose the array $ S$ is the $(\min,+)$-convolution of $A'$ and $B' $.  We call a sequence $\hat{S}$ an $(\alpha, \beta)$-\RS\ $(\min,
  +)$-convolution of $\hat{A}$ and $\hat{B}$ if $\hat{S}$ is a $\beta$-\RS\ which approximates the array $S$.
\end{defn}

By preserving an $(\alpha,\beta)$-\RS\ $(\min,+)$-convolution instead of an exact $(\min,+)$-convolution, we can reduce the storage space and the computation time.

\subsection{A fast algorithm for $(\alpha,\beta)$-\RS\ $(\min,+)$-convolution}
\label{sec:convalg}

Next, we give a simple algorithm \textsf{RSMinPlus} to compute an $(\alpha,\beta)$-\RS\ $(\min,+)$-convolution of two given $\alpha$-\RS s $\hat{A}$ and $\hat{B}$.

\topic{\textsf{RSMinPlus}$(\alpha, \beta, \hat{A}, \hat{B})$} We first compute the sum of every pair
$(\hat{a}[i_v], \hat{b}[j_t])$ where $\hat{a}[i_v] \in \hat{A}, \hat{b}[j_t] \in \hat{B}$. Define
$
\tilde{s}[l_r] =  \min_{ (i_v, j_t) \in \Phi_{l_r}} \hat{a}[i_v] + \hat{b}[j_t],
$
where  $\Phi_{l_r} = \{(i_v, j_t) \mid i_v+ j_t =  l_r ,  \hat{a}[i_v] \in \hat{A}, \hat{b}[j_t] \in \hat{B}   \} $. Suppose that there are $m$ different elements $\tilde{s}[l_r]$.
Then we sort $\tilde{s}[l_r]$ in the increasing order of the index number $l_r$. After sorting, we obtain a monotone increasing array $\tilde{S} = (\tilde{s}[l_1], \tilde{s}[l_2], \ldots, \tilde{s}[l_m]), l_r
< l_{r+1} \text{ for } r \in [m-1]$. Finally, we construct a $\beta$-\RS\ $\hat{S}$ from $\tilde{S}$
as our solution. Our construction is as follows.
\begin{enumerate}
\item Initially append $\hat{s}[l_m]=\tilde{s}[l_m]$ to $\hat{S}$. Let $\theta = \hat{s}[l_m]/(1+\beta)$.
\item Sequentially consider all elements in $\tilde{S}$ in decreasing order of the index number. If
$\tilde{s}[l_r] \leq \theta$, append $\hat{s}[l_r]=\tilde{s}[l_r]$ to $\hat{S}$. Let $\theta =
\hat{s}[l_r]/(1+\beta)$. Otherwise, ignore
$\tilde{s}[l_r]$ and consider the next element $\tilde{s}[l_{r-1}]\in \tilde{S}$.
\item Return the final sequence $\hat{S}$.
\end{enumerate}

\begin{lemma}
  \label{lm:min+}
  Suppose $\hat{A} = (\hat{a}[i_1], \hat{a}[i_2], \ldots, \hat{a}[i_{m_1}])$ and $\hat{B} =
  (\hat{b}[j_1], \hat{b}[j_2], \ldots,
  \hat{b}[j_{m_2}])$.   Let $m=\max\{m_1, m_2\}$. \textsf{RSMinPlus}$(\alpha, \beta, \hat{A},
  \hat{B})$ computes an $(\alpha,\beta)$-\RS\
  $(\min, +)$-convolution $\hat{S}$ of $\hat{A}$ and $\hat{B}$ in $O(m^2 \log m)$ time.
\end{lemma}

\begin{proof}
The correctness is not hard. Let $A'$ and $B'$ be the completions of $\hat{A}$ and $\hat{B}$
respectively. Let $S'$ be the exact $(\min, +)$-convolution of two arrays $A'$ and $B'$. By the
construction of $\tilde{S}$, we have that $S'$ is the completion of $\tilde{S}$.
Moreover, the $\beta$-\RS\ $\hat{S}$ approximates the array $S'$, which proves the correctness.
It remains to prove the running time.

By the algorithm \textsf{RSMinPlus}, it takes $m_1\cdot m_2=O(m^2)$ time to compute all
$\tilde{s}[l_r]$. Thus, there are at most $m^2$ different elements $ \tilde{s}[l_r] $ in $\tilde{S}$.
\eat{Observe that each index $l_r$ is an integer and its maximum value is $i_{m_1}+i_{m_2}=O(m)$. Therefore,
we can do the bucket sorting to obtain the array $\tilde{S}$ in $O(m^2)$ time.}Then we need
$O(m^2 \log m)$ time to sort all $\tilde{s}[l_r]$ and obtain the array $\tilde{S}$.
Finally, scanning all elements in $\tilde{S}$ to construct the $\beta$-\RS\ $\hat{S}$ needs $O(m^2)$ time.
Overall, the runtime of the algorithm is $O(m^2\log m)$.
\end{proof}

\topic{\textsf{FastRSMinPlus}. A more careful $(\alpha,\beta)$-\RS\ $(\min,+)$-convolution algorithm} In the tree sparsity model, we have some additional conditions which can improve
the running time of the algorithm \textsf{RSMinPlus}$(\alpha, \beta, \hat{A}, \hat{B})$. We consider the case that $\hat{A}$
are nonnegative sequences with each element $\hat{a}\in \hat{A}$ satisfying that either $\hat{a}=0$ or $\hat{a}\geq 1$. We have the same assumption on $\hat{B}$.
Moreover, $\alpha$ and $\beta$ are two constants such
that $0 < \beta \leq \alpha \leq 1$. The
intuition is as follows. Suppose
we have just appended some element $ \hat{s}[l] $  to the array $ \hat{S} $. By the property of
$(\alpha,\beta)$-\RS\ $(\min,+)$-convolution, we can safely ignore all $\tilde{s}[l_r]$ with
$\hat{s}[l]/(1+\beta)<\tilde{s}[l_r]\leq \hat{s}[l]$, and focus on finding the largest $
\tilde{s}[l_r]\in \tilde{S}$  such that $\tilde{s}[l_r] \leq
\hat{s}[l]/(1+\beta) $ (see the definition of $\tilde{S}$ in
\textsf{RSMinPlus}$(\alpha, \beta, \hat{A}, \hat{B})$).

We first build a hash table $\mathsf{HashB}$. The construction is as follows.
Each element in $\mathsf{HashB}$ is a pair $( key , value)$
where the $key $ term is an integer satisfying that $-1\leq key \leq \ceil{ \log_{1+\beta}
  \hat{b}[j_{m_2}] }$, and the $value$ term is the largest index $w$ satisfying that $\hat{b}[w]\in
\hat{B}$ and $\hat{b}[w]\leq (1+\beta)^{ key }$.
\footnote{If $\hat{b}[j_1]=0$, we define $\mathsf{HashB}(-1)=\hat{b}[j_1]$. Otherwise if $\hat{b}[j_1]\geq 1$, for a $key$ term, if we have
  $\hat{b}[j_1]>(1+\beta)^{key}$, we ignore this $key$ term when constructing $\mathsf{HashB}$.}
Symmetrically, we also build a hash table $\mathsf{HashA}$ for the array $\hat{A}$. Let
$U=\max\{\hat{a}[i_{m_1}],\hat{b}[j_{m_2}]\}$. Since both $\hat{A}$ and $\hat{B}$ are increasing
sequences, we can construct hash tables $\mathsf{HashA}$ and $\mathsf{HashB}$ in $O(\log_{1+\beta}
U)$ time by considering all $key$ terms in increasing order.

Given an element $\hat{s}[l]\in \hat{S}$, we show how to find the largest $\tilde{s}[l_r]\in \tilde{S}$  such that $\tilde{s}[l_r] \leq \hat{s}[l]/(1+\beta)$ by hash tables. We first reduce this problem to finding the element $\hat{b}[j_t]\in \hat{B}$ of the largest index $j_t$ for each $\hat{a}[i_v]\in \hat{A}$, such that $\hat{a}[i_v]+\hat{b}[j_t] \leq \hat{s}[l]/(1+\beta)$.
A simple scheme is to enumerate all
$\hat{a}[i_v] \in \hat{A}$, query the hash table $\mathsf{HashB}$, and find the largest $
\hat{b}[j_t]\in \hat{B}$ such that $\hat{b}[j_t] \leq \hat{s}[l]/(1+\beta)-\hat{a}[i_v]$. Then among
all such $(i_v,j_t)$ index pairs, we choose the pair $(i_{v^*},j_{t^*})$ with the largest sum
$i_{v^*}+j_{t^*}$. We append
$\hat{s}[i_{v^*}+j_{t^*}]=\tilde{s}[i_{v^*}+j_{t^*}]=\hat{a}[i_{v^*}]+\hat{b}[j_{t^*}]$ to
$\hat{S}$.  However, enumerating all elements is not necessary, since we have the following lemma.

\begin{lemma}
  \label{lm:half}
  Let $ \tau = \ceil{1/\alpha}$.
  For any $\hat{a}[i_{v-\tau}],\hat{a}[i_v]\in \hat{A}$, we have $ \hat{a}[i_{v - \tau }] \leq \hat{a}[i_v ]/2 $. Similarly, for any $\hat{b}[j_{t-\tau}],\hat{b}[j_t]\in \hat{B}$,  $ \hat{b}[ j_{t - \tau}] \leq \hat{b}[j_t ]/2 $.
\end{lemma}

\begin{proof}
  W.l.o.g., we only consider the array $\hat{A}$.
  By Definition~\ref{def:sparsemin}, we know $(1+\alpha) \hat{a}[i_{v-1}] \leq \hat{a}[i_v]$. If $\alpha \geq
  1$, $ \tau= 1$, the lemma is trivially true. Otherwise if $\alpha < 1$, we have $\hat{a}[i_{v-
    \tau}] \leq  \hat{a}[i_v] / (1+\alpha)^{1/\alpha} \leq \hat{a}[i_v]/2$.
\end{proof}

Let $\theta=\hat{s}[l]/(1+\beta)$. Assume that $ \hat{a}[i_v]\in \hat{A}  $ (resp. $ \hat{b}[j_t]
\in \hat{B} $) is the
largest element such that $\hat{a}[i_v] \leq \theta$ (resp. $\hat{b}[j_t] \leq \theta$). Hence, both
$\hat{a}[i_{v+1}]$ and $\hat{b}[j_{t+1}]$ are at least $\theta  $.
Therefore, the pair $(i_{v^*},j_{t^*})$ must satisfy that either $ v-\tau \leq v^*\leq v $ or
$t-\tau\leq t^* \leq t$, since $\hat{a}[i_{v-\tau}] + \hat{b}[j_{t-\tau}] \leq \theta$ by Lemma
\ref{lm:half}. Thus, we only need to consider at most $\tau +1$ elements in $\hat{A}$ or
$\hat{B}$. Note that we can directly find such index $i_v$ (resp. $ j_t $) by the following lemma.

\begin{lemma}
  \label{lm:collision}
  Either $i_v=\mathsf{HashA} (\ceil{\log_{1+\beta}\theta})$ or $i_v=\mathsf{HashA}
  (\floor{\log_{1+\beta}\theta})$. Similarly, either $j_t =\mathsf{HashB}
  (\ceil{\log_{1+\beta}\theta}$ or $j_t =\mathsf{HashB} (\floor{\log_{1+\beta}\theta}$
\end{lemma}
\begin{proof}
  W.l.o.g., we take $i_v$ as example. Let $i_w=\mathsf{HashA} (\ceil{\log_{1+\beta}\theta}) $. If
  $i_v\neq i_w $, then by the definition of $i_v$, we have that
  $ \theta< \hat{a}[i_w]\leq (1+\beta)^{\ceil{\log_{1+\beta}\theta}}$. Since $\hat{A}$ is an
  $\alpha$-\RS\ and $\alpha\geq \beta$, we have
  $$ \hat{a}[i_{w-1}]\leq \hat{a}[i_w]/(1+\alpha)\leq \hat{a}[i_w]/(1+\beta)\leq
  (1+\beta)^{\ceil{\log_{1+\beta}\theta}-1}\leq (1+\beta)^{\floor{\log_{1+\beta}\theta}}\leq \theta
  $$
  Thus, we have $i_v=i_{w-1}$ by the definition of $i_v$. Note that $i_{w-1}=\mathsf{HashA}
  (\floor{\log_{1+\beta}\theta})$. We finish the proof.
\end{proof}

\begin{algorithm}[htp]
  \caption{FastRSMinPlus$(\alpha, \beta,\hat{A},\hat{B})$}
  \label{alg:fastmin+}
  \SetKwFunction{HashA}{HashA}
  \SetKwFunction{HashB}{HashB}
  \KwData{ $0\leq \beta\leq \alpha\leq 1, \hat{A}=(\hat{a}[i_1],\hat{a}[i_2],\ldots,\hat{a}[i_{m_1}]),  \hat{B}= (\hat{b}[j_1],\hat{b}[j_2],\ldots,\hat{b}[j_{m_2}])$ }
  \KwResult{$\hat{S}$}
  Initialize: $\tau=\ceil{1/\alpha}$,
  $ \hat{S} \leftarrow \{\hat{s}[i_{m_1}+j_{m_2}] = \hat{a}[i_{m_1}]+\hat{b}[j_{m_2}]\}$,  $ \theta = \hat{s}[i_{m_1}+j_{m_2}]/(1+\beta)$  \;

  For $-1\leq key \leq \ceil{ \log_{1+\beta} \hat{a}[i_{m_1}] }$, let $value$ be the largest index $w$ satisfying that $\hat{a}[w]\in \hat{A}$ and $\hat{a}[w]\leq (1+\beta)^{ key }$. Let \HashA\ be the collection of these $(key,value)$ pairs \;
  For $-1\leq key \leq \ceil{ \log_{1+\beta} \hat{b}[j_{m_2}] }$, let $value $ be the largest index $w$ satisfying that $\hat{b}[w]\in \hat{B}$ and $\hat{b}[w]\leq (1+\beta)^{ key }$. Let \HashB\  be the collection of these $(key,value)$ pairs \;

  \While{$ \theta > (\hat{a}[i_{1}] + \hat{b}[j_{1}]) $}{
  $\theta'\leftarrow \theta-\hat{a}[i_1]$,
    $ j_{t_1} \leftarrow $ \HashB{$\floor{\log_{1+\beta} \theta'}$},
    $ j_{t_2} \leftarrow $ \HashB{$\ceil{\log_{1+\beta} \theta'}$} \;
    If $\hat{b}[j_{t_2}]\leq \theta'$, let $j_t\leftarrow j_{t_2}$. Otherwise, let $j_t\leftarrow j_{t_1}$ \;
    $  l \leftarrow j_t+i_1 , \hat{s}[l] \leftarrow \hat{b}[j_t]+\hat{a}[i_1] $ \;
    \For { $ \Delta = 1$ \emph{\KwTo} $ \min\{ t, \tau \}$}{
      $ \delta \leftarrow \theta - \hat{b}[j_{t-\Delta}] $ \;
      Let $ w  \leftarrow $ \HashA {$\ceil{\log_{1+\beta}\delta}$ }. If $\hat{a}[w]>\delta$, let $ w  \leftarrow $ \HashA {$\floor{\log_{1+\beta}\delta}$ }\;
      \If {$ l<w+j_{t-\Delta}  $, \emph{or} $(l=w+j_{t-\Delta}  $
        \emph{and} $\hat{s}[l] > \hat{a}[w] +\hat{b}[j_{ t -\Delta}])$ }{
        $ l\leftarrow w+ j_{t -\Delta} , \hat{s}[l] \leftarrow \hat{a}[w] + \hat{b}[j_{t -\Delta}]   $ \;
      }
    }

    $\theta'\leftarrow \theta-\hat{b}[j_1]$, $ i_{v_1} \leftarrow $ \HashA{$\floor{\log_{1+\beta} \theta' }$},
    $ i_{v_2} \leftarrow $ \HashA{$\ceil{\log_{1+\beta} \theta'}$} \;
    If $\hat{a}[i_{v_2}]\leq \theta'$, let $i_v\leftarrow i_{v_2}$. Otherwise, let $i_v\leftarrow i_{v_1}$ \;
    \For { $ \Delta = 0$ \emph{\KwTo} $ \min\{ v, \tau \}$}{
      $ \delta \leftarrow \theta - \hat{a}[i_{v-\Delta}] $ \;
      Let $ w  \leftarrow $ \HashB {$\ceil{\log_{1+\beta}\delta}$ }. If $\hat{b}[w]>\delta$, let $ w  \leftarrow $ \HashB {$\floor{\log_{1+\beta}\delta}$ }\;
      \If {$ l<w+i_{v-\Delta}  $, \emph{or} $(l=w+i_{v-\Delta}  $
        \emph{and} $\hat{s}[l] > \hat{b}[w] +\hat{a}[i_{ v -\Delta}])$ }{
        $ l\leftarrow w+ i_{v -\Delta} , \hat{s}[l] \leftarrow \hat{b}[w] + \hat{a}[i_{v -\Delta}]   $ \;
      }
    }

    $ \hat{S}\leftarrow \hat{s}[l] \cup \hat{S} $, $\theta =\hat{s}[l]/(1+\beta) $\;

  }
  \KwRet{$\hat{S}$} \;
\end{algorithm}

We summarize our approaches in Algorithm~\ref{alg:fastmin+}. The analysis of the algorithm is as follows.

\begin{lemma}
  \label{lm:fastpre}
  Let $U=\max\{\hat{a}[i_{m_1}],\hat{b}[j_{m_2}]\}$. Let $0\leq \beta\leq \alpha \leq 1$ be two constants. The algorithm $\mathsf{FastRSMinPlus}(\alpha, \beta,\hat{A},\hat{B})$
  computes an $(\alpha,\beta)$-\RS\ $(\min, +)$-convolution $\hat{S}$ of $\hat{A}$ and $\hat{B}$ in
  $O \left( \frac{\log_{1+\beta} U }{\alpha}\right)  $ time.
\end{lemma}

\begin{proof}

  We first prove the correctness. Let $A$ and $B$ be the completions of $\hat{A}$ and $\hat{B}$ respectively. Let $S$ be the $(\min, +)$-convolution of $A$ and $B$. Assume that we just append an element $\hat{s}[l]$ to $\hat{S}$ and let $\theta=\hat{s}[l']/(1+\beta)$. Consider the next recursion from Line 5 to Line 20, we append a new element $\hat{s}[l]$ to $\hat{S}$. Initially in Line 6, we find the largest index $j_t$ satisfying that $\hat{b}[j_t]+\hat{a}[i_1]\leq \theta$ by Lemma \ref{lm:collision}. We first analyse the first loop from Line 8 to Line 12. In Line 10, we find the largest element $\hat{a}[w]$ satisfying that $\hat{a}[w]\leq \theta-\hat{b}[j_{t-\delta}]$ by Lemma \ref{lm:collision}.
  Thus, we conclude that $\hat{a}[w]+\hat{b}[j_{t-\Delta}]\leq \theta$ for any $1\leq \Delta\leq
  \min\{t,\tau\}$. Similarly, we can prove that $\hat{b}[w]+\hat{a}[i_{v-\Delta}]\leq \theta$ for
  any $1\leq \Delta\leq \min\{v,\tau\}$ in Line 17. Thus, the element $\hat{s}[l]$ always satisfies
  that $\hat{s}[l]\leq \theta=\hat{s}[l']/(1+\beta)$ during the recursion. Thus, the output
  $\hat{S}$ is a $\beta$-\RS.

  On the other hand, let $S'$ be the completion of $\hat{S}$. Assume that
  $s[l]=\hat{a}[i_{v^*}]+\hat{b}[j_{t^*}]$ is the largest element in $S$ satisfying that $s[l]\leq
  \theta$. W.l.o.g. we assume that $\theta/2\leq \hat{b}[j_{t^*}]\leq \theta-\hat{a}[i_1]$
  (otherwise $\theta/2\leq \hat{a}[i_{v^*}]\leq \theta-\hat{b}[j_1]$). We conclude that $t-\tau \leq
  t^*\leq t$ by Lemma \ref{lm:half}. Then we must consider the element $\hat{b}[j_{t^*}]$ in the
  loop from Line 8 to Line 12. Note that in Line 10, we find an index $w= i_{v^*}$ by Lemma
  \ref{lm:collision}. By the updating rules in Line 11-12, we update
  $\hat{s}[l]=\hat{a}[i_{v^{*}}]+\hat{b}[j_{t^*}]$ in Line 12 and append $\hat{s}[l]$ to $\hat{S}$
  in Line 20.
  Considering any element $s'[l_0]$ with $l+1\leq l_0\leq l'$, we have that $s'[l_0]=\hat{s}[l']$ by
  Definition \ref{def:extended}. Moreover, we have the following inequality by the chosen of $l_r$,
  $$ s[l_0]\leq s'[l_0] =\hat{s}[l']=s[l'] =(1+\beta) \theta < (1+\beta) s[l+1]\leq (1+\beta) s[l_0].
  $$
  Overall, we prove that $S'$ is a $\beta$-approximation of $S$ by Definition \ref{def:seqapp}.

  Then we analyze the running time. By the definition of $U$, we always have $\theta \leq 2U$. After
  each iteration, the value $\theta$ decreases by a factor at least $1+\beta$ by the fact that
  $\hat{S}$ is a $\beta$-\RS. Thus, there are at most $\ceil{\log_{1+\beta} 2U }$ iterations. For
  each iteration, we first find $j_t$ in $O(1)$ time by Lemma \ref{lm:collision}. Then we consider
  at most $\tau+1=\ceil{1/\alpha}+1$ possible index pairs $(w,j_{t-\Delta})$. We only cost $O(1)$
  time for each index pair. For the loop from Line 13 to Line 19, we have the same analysis. Thus,
  the running time of each iteration is $O(\ceil{1/\alpha})$. Overall, the total running time is at
  most $O ( \frac{\log_{1+\beta} U }{\alpha}) $.
\end{proof}

By Lemma \ref{lm:fastpre}, the running time of Algorithm \ref{alg:fastmin+} is determined by the
term  $U=\max\{\hat{a}[i_{m_1}], $ $ \hat{b}[j_{m_2}]\}$. In fact, if $\log_{1+\beta} U$ is larger than
the largest index number $M=\max \{i_{m_1},j_{m_2}\}$ of arrays,  we can improve the running time of
Algorithm~\ref{alg:fastmin+} further. The main difference is that we do not use hash tables since it
takes $\log_{1+\beta} U$ time for construction. The details are as follows.

\begin{enumerate}
\item Compute the completion
$A'=(a'[0], a'[1], \ldots, a'[i_{m_1}])$ and $B'=(b'[0],b'[1],\ldots,b'[j_{m_2}])$ of $\hat{A}$ and $\hat{B}$ respectively.
\item Compute
the $(\min, +)$-convolution $S$ of $A'$ and $B'$ as follows. Let
$\tau=\ceil{1/\alpha}$. Sequentially consider each $L\in [0, i_{m_1}+j_{m_2}]$ in the increasing order. For a term $L$, find $\hat{a}[i_v]\in \hat{A}$ with the largest index satisfying that $i_v\leq L$. Similarly, find $\hat{b}[j_t]\in \hat{B}$ with the largest index satisfying that $j_t\leq L$.
\item  Compute
  $ s[L] = \min \{ \min_{ 0\leq \Delta \leq  \min\{v,\tau\}}(a'[i_{v-\Delta}]+b'[L-i_{v-\Delta}] ),
  \min_{0\leq \Delta \in \leq  \min\{t,\tau\}} (a'[L-j_{t-\Delta}]+b'[j_{t-\Delta}] )\} $.
\item Scan the array $S$ in decreasing order. Construct a $\beta$-\RS\ as in Algorithm \textsf{RSMinPlus}$(\alpha, \beta, \hat{A}, \hat{B})$
\end{enumerate}

Note that our approach is similar to the iteration of Algorithm \textsf{RSMinPlus}$(\alpha, \beta,
\hat{A}, \hat{B})$. While we use the properties of $\alpha$-\RS\ during computing $S$, similar
to Algorithm~\ref{alg:fastmin+}.
The running time of Step 1 is $O(M)$. For each $L$, since we consider $L$ sequentially, it costs
$O(1)$ time to find indexes $i_v$ and $j_t$. Moreover, we cost $O(\tau)$ time to compute
$s[L]$. Finally, the running time of Step 4 is $O(M)$. Thus, the total running time is $ O(M/\alpha)
$. Combining with Lemma \ref{lm:fastpre}, we have the following lemma.

\begin{lemma}
  \label{lm:fastmin+}
  Consider two $\alpha$-\RS s $\hat{A}=(\hat{a}[i_1], \hat{a}[i_2],\ldots, \hat{a}[i_{m_1}])$ where each $\hat{a}[i_w]$ ($1\leq w\leq m_1$) satisfies that either $\hat{a}[i_w]=0$ or $\hat{a}[i_w]\geq 1$, and $\hat{B}= (\hat{b}[j_1], \hat{b}[j_2],
  \ldots, \hat{b}[j_{m_2}])$ where each $\hat{b}[j_w]$ ($1\leq w\leq m_2$) satisfies that either $\hat{b}[j_w]=0$ or $\hat{b}[j_w]\geq 1$. Let $U = \max\{\hat{a}[i_{m_1}] , \hat{b}[j_{m_2}]\}$ and $M =\max\{ i_{m_1}, j_{m_2}\}$. Let $0\leq \beta\leq \alpha \leq 1$ be two constants. There exists an algorithm $\mathsf{FastRSMinPlus}(\alpha, \beta,\hat{A},\hat{B})$
  computing an $(\alpha,\beta)$-\RS\ $(\min, +)$-convolution $\hat{S}$ of $\hat{A}$ and $\hat{B}$ in
  $O \left(\min\left\{ \frac{\log_{1+\beta} U }{\alpha}, \frac{M}{\alpha}\right\}\right)  $ time.
\end{lemma}

\section{Tree Sparsity Model}
\label{sec:tree}

In this section, we discuss the tree sparsity model. We first introduce some essential definitions
and techniques such as weight discretization and \RS\ $(\min, +)$-convolution. Using these new
techniques, we will give an $O(\e^{-1}n\log n)$ time algorithm. Then we speed up the algorithm
to $O(\e^{-1}n(\log\log\log n)^2)$ time through a faster algorithm for \RS\ $(\min, +)$-convolution
and the lookup table method. Our improved algorithm is appropriate for both \headtree\ and
\tailtree. Moreover, we show that we can obtain a linear time algorithm for \headtree\ by a more
careful weight discretization technique. For \tailtree, the new weight discretization technique is
not suitable.  Instead, we give a linear time algorithm for \tailtree\ under the assumption that $k
\leq n^{1-\delta}$ ($\delta \in (0, 1]$ is a fixed constant) by a pruning technique.

 For convenience, we only consider the perfect binary tree in this section. Our algorithm can be
naturally extended to the general complete $b$-ary tree sparsity model. We defer the details in
Appendix~\ref{sec:extend}.  In this section, we only consider the $l_1$-norm for both
\headtree\ and \tailtree. Hence, we only consider the case that each node weight $x_i\geq 0$. We
will see our algorithm can be easily  generalized to general $l_p$-norm. Again, We defer the details in
Appendix~\ref{sec:extend}.

We denote the given perfect binary tree by $T$. Consider a node in the tree $T$.  Suppose
the number of edges on the path between the node and the root is $t$. Define the level of the node by $(\log (n+1)
-t) $. For example, all leaves are at level  $1$ and the root node is at level $\log (n+1)$.
\footnote{Note that $\log (n+1)$ is an integer since $T$ is a perfect binary tree.}
For each level of $T$, we sort all nodes in the same level by a BFS. We denote by
$N_{ij}$ the $j$th node at level $i$. We call the subtree with root $N_{ij}$ and containing all
nodes rooted at $N_{ij}$ the \emph{largest subtree} of $N_{ij}$ and denote it by $T_{ij}$.
Note that the left child of $N_{ij}$ is $N_{i-1,2j-1}$
and the right child is $N_{i-1,2j}$.

Assume that each node $N_{ij}$ has a weight $x_{ij}$.   Recall that in the tree sparsity model,
each support $\Omega \in \mathbb{T}_k$ is a subtree of $T$ rooted at the root node $N_{\log (n+1), 1}$ with $k$ nodes.  For a node $N_{ij}$ and a subtree $\Omega\in \mathbb{T}_k$, we use $N_{ij} \notin \Omega$ to
denote $N_{ij} \in T\setminus \Omega$. In this section, we first consider the \tailtree\ version. The
\headtree\ version is similar to \tailtree, and we will show the differences later. We denote the
optimal solution of the \tailtree\ problem by $\Omega^*$ together with an optimal \emph{tail value}
$\OPT=\sum_{N_{ij} \notin \Omega^*} x_{ij}$. We also denote the solution of our algorithm by
$\hat{\Omega}$ together with a tail value $\SOL=\sum_{N_{ij} \notin \hat{\Omega}} x_{ij}$. W.l.o.g., we
assume that $k\geq \log n$. Otherwise we can safely ignore those nodes $N_i$ of depth larger than
$k$. We also consider the error parameter $\e>0$ as a constant.

\subsection{A Nearly Linear Time Algorithm for \tailtree}
\label{sec:nearlinear}

We first propose a scheme for the tail-approximation projection problem for the general case.
We first assume that each node weight $x_{ij}$ is an integer among  $[0, \frac{n \log
  n}{\e}+n]$. Thus
there are at most $O(n\log n / \epsilon)$ different weight values. We can remove this assumption by a weight discretization technique, see Appendix \ref{sec:weight} for details. We then introduce a look-up table
method, which is inspired by the well known  \emph{Four Russians Method}\cite{kronrod1970economic}.
Combining \textsf{FastRSMinPlus} and the look-up table method, we give a nearly linear time algorithm for \tailtree.

\eat{
  However, it just preserve approximate $(\min,
  +)$-convolution. Tolerating inexactness help us to reduce the data size. In the exact algorithm,
  each node should maintain $k$ states. Instead, we construct a data structure such that each node
  maintains only $O(\log n/ \e)$ states. The data structure can efficiently reduce the running time.
}

\topic{Encoding low levels by the look-up table method}
In fact, we can further discretize the
weight such that there are at most $O( \log n / \e)$ different discretized weight.
Define $\hat{x}_{ij} = (1+\e)^{\ceil{
    \log_{1+\e} x_{ij}}} $ as the discretized weight of node $N_{ij}$. Therefore, $ x_{ij} \leq
\hat{x}_{ij}  <  (1+\e)x_{ij}   $. Suppose that $ s[k] = \sum_{N_{ij} \notin \Omega^*} x_{ij}$ is the optimal tail value for \tailtree, where $\Omega^*$ is
the optimal support using node weights $\{x_{ij}\}$. Suppose that $\hat{\Omega}$ is the optimal support for \tailtree\ using discretized weights $\{\hat{x}_{ij}\}$. We have the following inequality
$$
  \sum_{N_{ij} \notin \hat{\Omega}} x_{ij} \leq \sum_{N_{ij} \notin \hat{\Omega}} \leq \sum_{N_{ij} \notin \Omega^*} \hat{x}_{ij} \leq (1+\e) \sum_{N_{ij}
    \notin \Omega^*} x_{ij} = s[k].
$$

Thus, we use the discretized weights in the following. Now we have at most $O(\log_{1+\e}( n \log n /\e) )= O( \log n /\e)$ different weights.
Consider any node $N_{\xi j}$ at level $  \xi =\ceil{\log\log n - \log(1/\e) - \log\log\log n}$. The
largest subtree $T_{\xi j}$ rooted at $N_{\xi j}$ has at most $m = \ceil{\log n/ (\e \log\log n)}$ nodes. We can compute its exact tail array with running time at most $\sum_{i \in [1,\xi]} (m+1)\cdot 2^{-i} \cdot 2^{2i} = O(\log^2 n/\e^2)$ by computing exact
$(\min,+)$-convolution level by level. Since we have at most $O( \log n /\e)$ different node weights after discretization, there are at most $ O( \log n / \e)^{m} = O( n^{O(\e)}) $
possible constructions for $T_{\xi j}$.
\footnote{Here, each construction is a weight assignment of all nodes in $T_{\xi j}$.}
By this observation, we can enumerate all possible constructions and compute the corresponding exact tail array using $ O( n^{O(\e)} \log^2 n/ \e^2) =o(n) $ time and $o(n)$ space. Thus, we encode all possible constructions of subtrees at level $\xi$
into a look-up table. When we need compute the exact tail array of any node at level $ \xi $, we search the
look-up table and return the array in $O(m)$ time.

\vspace{0.1in}

Now we are ready to give our algorithm for \tailtree. For
each node $N_{ij}$, define an array $S_{ij}=(s[0],s[1],s[2],\ldots,s[2^i-1])$ to be the exact tail
array, where each element $s[l]$
represents the optimal tail value for \tailtree\ on $T_{ij}$, i.e., $s[l]=\min_{\Omega\in
  \mathbb{T}_{2^i-1-l} (T_{ij})}\sum_{N_{i'j'}\in T_{ij} \setminus \Omega} x_{i'j'}$. In the exact algorithm, we in fact
compute the exact tail array $S_{ij}$ for each node $N_{ij}$ through the
$(\min, +)$-convolution. Our main technique is to maintain an $\alpha$-\RS\  $\hat{S}_{ij}$ for
each $S_{ij} $. The value of $\alpha$ depends on the level $i$, which will be decided later.

\topic{\textsf{FastTailTree}} In our algorithm, we use $\mathsf{MinPlus}$ to
represent the $O(m^2/2^{c\sqrt{\log m}})$ algorithm for exact $(\min,+)$-convolutions mentioned in
\cite{williams2014faster} ($c>0$ is some fixed constant). We divide the whole tree $T$
into three \layers\ as follows.
\begin{description}
\item[Step 1:] Let $\xi = \ceil{(\log\log n -
\log(1/\e) - \log\log\log n)}, \eta = \ceil{(\log\log n +\log (1/\e))}$. For any node $N_{\xi j}$ at level $ \xi $, we use the look-up table method to
  obtain the exact tail array $S_{\xi j}$. Compute an $\e'$-\RS\ $\hat{S}_{\xi j}$ which approximates $S_{\xi j}$, where $\e' = \e/(\eta-\xi+1)$.

\item[Step 2:] From level $i=\xi+1$ to level $i=\eta$, we use $\mathsf{FastRSMinPlus}(\e',\e',\hat{S}_{i-1,2j-1},\hat{S}_{i-1,2j})$
  to compute a sequence $\hat{S}_{ij}$ for any node $N_{i j}$, and append $\hat{s}[2^i-1]= (1+\e')^{i-\xi+1} \cdot \sum_{N_{i'j'}\in T_{ij}}x_{i'j'}$ to $\hat{S}_{ij}$.
  \footnotetext{We set $\hat{s}[2^i-1]$ to be this value, because we want to guarantee that $\hat{S}_{ij}$ is still an $\e'$-\RS\ after appending $\hat{s}[2^i-1]$. This is convenient for analyzing the algorithm in Theorem \ref{thm:simtail}.}
  Compute an $\e$-\RS\ $\hat{S}_{\eta j}$ which approximates $S_{\eta j}$ for any node $N_{\eta j}$.

\item[Step 3:] From level $i=\eta+1$ to level $i=\log (n+1)$, we use $ \mathsf{FastRSMinPlus}(\e_{i-1},\e_i,\hat{S}_{i-1,2j-1},\hat{S}_{i-1,2j})$
  to compute a sequence $\hat{S}_{ij}$ for any node $N_{i j}$, where $\e_i = \frac{\e}{ 3^{(i- \eta)/4}}$ ($\eta \leq i\leq \log (n+1)$). We then append $\hat{s}[2^i-1]= (1+\e')^{\eta-\xi+1} \cdot \prod_{l=\eta}^{i} (1+\e_{l}) \cdot \sum_{N_{i'j'}\in T_{ij}}x_{i'j'}$ to $\hat{S}_{ij}$.

\item[Step 4:] Let $\hat{s}[L]\in \hat{S}_{\log (n+1),1}$ be the smallest element of index $L$ satisfying that $L\geq n-k$. Output $\hat{\Omega}\leftarrow$ \FindTree{$L, T$}. Here, $\mathsf{FindTree}$ is a
backtracking process with running time $O(n)$ which obtains a feasible solution $\hat{\Omega}$. We defer the details
in Algorithm~\ref{alg:findtree} in Appendix~\ref{sec:extend}.

\end{description}

\vspace{0.1in}

Before analyzing \textsf{FastTailTree}, we give some intuitions about why we compute
$(\alpha,\beta)$-\RS\ $(\min, +)$-convolutions.
Note that the weight $x_i$ of each node is an integer at most $ O(n\log n/\e)$.
Thus, the maximum value in each $S_{ij}$ is at most $O(n^2\log
n / \e)$. In our algorithm, we use a sequence $\hat{S}_{ij}$ to approximate $S_{ij}$.
By Definition~\ref{def:seqapp} and~\ref{def:sparsemin}, the number of elements in $\hat{S}_{ij}$ is
at most $ \log_{(1+\e_i)} (n^2\log n /\e) = O(\log n/ \e_i)$,
which means that the size of
$\hat{S}_{ij}$ is sublinear on $n$. Thus, if the level $i$ is high enough, the array $\hat{S}_{ij}$
maintains much fewer elements than $S_{ij}$, and can be constructed faster.
By Lemma \ref{lm:fastmin+}, we have the following corollary.

\begin{corollary}
  \label{cor:fastmin+} At Step 2, $ \mathsf{FastRSMinPlus}(\e',\e',\hat{S}_{i-1,2j-1},\hat{S}_{i-1,2j})$ can be computed in time
 $O(\frac{2^{i-1}}{\e'})$. At Step 3, $ \mathsf{FastRSMinPlus}(\e_{i-1},\e_i,\hat{S}_{i-1,2j-1},\hat{S}_{i-1,2j})$ can be computed in time
  $O(\frac{\log n}{\e_{i-1}\e_{i}})$.
\end{corollary}

\begin{proof}
  For each node $N_{ij}$, the largest index of $\hat{S}_{i-1,2j-1}$ or $\hat{S}_{i-1,2j}$ is at most
  $M\leq 2^{i-1}$. On the other hand, the maximum value $U$ in $\hat{S}_{i-1,2j-1}$ or
  $\hat{S}_{i-1,2j}$ is at most $O(n^2\log n / \e)$. By Lemma \ref{lm:fastmin+}, we prove the
  corollary.
\end{proof}

Now we are ready to give the following main theorem.

\begin{thm}
  \label{thm:simtail}
  Algorithm \textsf{FastTailTree} is a $ (1+ \e)$-approximation algorithm with running time $ O( \e^{-1}n(\log\log\log n)^2 ) $
  for \tailtree.
\end{thm}

\begin{proof}
  We first prove the running time. There are $(n+1)/2^{\xi}$ nodes at level $\xi$. We
need $O(2^{\xi})$ time to compute each exact tail array through searching the look-up table and need $O(2^{\xi})$ time to compute an $\e'$-\RS. Thus,
the  runtime of Step 1 is $O( 2^{-\xi}(n+1) \cdot 2^\xi) = O(n)$. Considering Step 2, the running time for each node $N_{ij}$ is at most $O(2^i/\e')$. Thus, the total time of Step 2 is
$$
  \sum_{i \in (\xi, \eta  ]} \frac{n+1}{2^i}\cdot \frac{ 2^{i}}{\e'} =O ( \e^{-1}n((\log\log\log n)^2 + \log^2(1/\e) )).
$$
For Step 3, the running time for each node $N_{ij}$ is at most $O(\log n/ (\e_{i}\e_{i-1}))  $ by Corollary~\ref{cor:fastmin+}.  Thus, the total time  of Step 3 is
$$
  \sum_{i\in (\eta ,\log n ]} \frac{n+1}{2^i}\cdot \frac{\log  n}{\e_{i-1}\e_{i}} = O(\e^{-1} n).
$$
The running time of Step 4 is $O(n)$. Overall, the total running time is  $O(n + \e^{-1}n(\log\log\log n)^2 +\e^{-1}n) = O( \e^{-1}n(\log\log\log n)^2 )  $.

  Then we prove the correctness by showing that
  our solution $\hat{\Omega}$ is a $(1+\e)$-approximation for the optimal solution $\Omega^*$.
  We first prove by induction that for each node $N_{ij}$ at level $\xi\leq i\leq \eta$, the array $\hat{S}_{ij}$ is a $\bigl((1+\e')^{i-\xi+1}-1\bigr)$-\RS\ which approximates the exact tail array $S_{ij}$.
  The base case at level $\xi$ is true since each sequence $\hat{S}_{\xi j}$ is an $\e'$-\RS\ of the exact tail array $S_{\xi j}$ at Step 1. 
   Then we suppose that for level $i-1$ ($\xi+1\leq i\leq \eta$), any sequence $\hat{S}_{i-1,j}$ is an $\e^*=\bigl((1+\e')^{i-\xi}-1\bigr)$-\RS\ which approximates the array $S_{i-1,j}$. We consider an arbitrary node $N$ and its sequence $\hat{S}$ at level $i$. Let $S'_1$ be the completion of the sequence $\hat{S}_{1}$ maintained by $N$'s left child. Let $S'_2$ be the completion of the sequence $\hat{S}_{2}$ maintained by $N$'s right child. Let $S_1$ and $S_2$ be the exact tail arrays of $N$'s left and right children respectively. Let $S$ be the exact $(\min,+)$-convolution of $S_1$ and $S_2$, i.e., $S$ is the exact tail array of $N$ without the last term $s[2^i-1]$.
  By induction, we know that the two arrays $S'_1$ and $S'_2$ are $\e^*$-approximations of $S_1$ and $S_2$ respectively.
  Let $\tilde{S}$ be the exact $(\min,+)$-convolution of $S'_1$ and
  $S'_2$. Let $S'$ be the completion of $\hat{S}$ (without the element $\hat{s}[2^i-1]$).
  By Definition \ref{def:seqapp} and~\ref{def:sparsemin}, we have that $S'$ is
  an $\e'$-approximation of $\tilde{S}$.

  Consider any element $\tilde{s}[l]
  \in \tilde{S}$ such that $\tilde{s}[l] = a'[l_1] + b'[l_2]$ for $l_1+l_2 =l, a'[l_1]\in S'_{1} ,
  b'[l_2] \in  S'_{2}$. By induction, we have that $ a[l_1] \leq a'[l_1] $ for $a[l_1]\in S_1$ and $
  b[l_2]\leq
  b'[l_2] $ for $b[l_2]\in S_2$. Therefore, we have that
  $$s[l]\leq a[l_1]+b[l_2] \leq a'[l_1] +  b'[l_2]=\tilde{s}[l] \leq s'[l].$$
  The last inequality follows from the fact that $S'$ is an $\e'$-approximation of $\tilde{S}$. On
  the other hand, consider any element $s[l]\in S$ such that $s[l]=a[l_1]+b[l_2]$ for $l_1+l_2 =l,
  a[l_1]\in S_{1} , b[l_2] \in  S_{2}$. By induction, we have that $ a'[l_1] \leq (1+\e') a[l_1] $
  for $a'[l_1]\in S'_1$ and $ b'[l_2]\leq (1+\e^*)b[l_2] $ for $b'[l_2]\in S'_2$. Thus, we conclude
  that
  $$ s'[l]\leq (1+\e')\tilde{s}[l] \leq (1+\e')(a'[l_1]+b'[l_2])\leq (1+\e')(1+\e^*)(a[l_1]+b[l_2]) =s[l].
  $$
  By the above argument, $S'$ is a $\bigl((1+\e')(1+\e^*)-1\bigr)$-approximation of $S$. More
  specifically, we have the following inequality
  $$s'[2^i-2]=a'[2^{i-1}-1]+b'[2^{i-1}-1]\leq (1+\e^*)(a[2^{i-1}]+b[2^{i-1}-1])=(1+\e^*)s[2^{i}-2].
  $$
  Now we consider the element $\hat{s}[2^i-1]$ appended to $\hat{S}$ at Step 2. On one hand, since the exact tail value $s[2^i-1]=\sum_{N_{i'j'}\in T_{ij}}x_{i'j'}$, we have that $s[2^i-1]\leq \hat{s}[2^i-1]\leq (1+\e')(1+\e^*)s[2^i-1]$. On the other hand, we have
  $$(1+\e')\hat{s}[2^i-2]=(1+\e')s'[2^i-2]\leq (1+\e')(1+\e^*)s[2^i-2]\leq (1+\e')(1+\e^*)s[2^i-1]\leq \hat{s}[2^i-1].$$
  The last inequality follows from the fact that $\hat{s}[2^i-1]=(1+\e')(1+\e^*)s[2^i-1]$. Thus, we conclude that $\hat{S_{ij}}$ is still an $\e'$-\RS\ approximating the exact tail array $S_{ij}$, which proves the induction. 
  
  By a similar reduction, we can prove that $\hat{S}_{\log (n+1),1}$ is a $\bigl((1+\e')^{\eta-\xi+1} \cdot \prod_{l=\eta}^{\log (n+1)} (1+\e_{l})-1\bigr)$-\RS\ which approximates $S_{\log (n+1),1}$. Overall, the approximation ratio for the root array $\hat{S}_{\log(n+1),1}$ is
 $1+ O(\e)$.
  Therefore, let $S'_{\log(n+1),1}$ be the completion of the sequence $\hat{S}_{\log(n+1),1}$
  maintained in the root node. Let $s'[n-k]\in S'_{\log(n+1),1}$, we have that $\hat{s}[L]\leq
  (1+\e_{\log (n+1)}) s'[n-k]\leq (1+O(\e))s[n-k]$ for $s[n-k]\in S_{\log (n+1),1}$. By using a small enough value $\theta(\e)$ to replace $\e$, we can
  guarantee that the value $\hat{s}[L]$ is a $(1+\e)$-approximation tail value for \tailtree.
\eat{
  By the backtracking process $\mathsf{FindTree}(L, T)$, we obtain a support $\hat{\Omega}$ with a
  tail value at most $\hat{s}[L]$ since each element $\hat{s}[L]$ is at least as large as the exact
  tail value $s[L]$ by the algorithm. On the other hand, $|\hat{\Omega}|\leq k$ since $L\geq n-k$.
}

\end{proof}

\subsection{A Linear Time Algorithm for \tailtree\ if $k \leq n^{1-\delta}$}

For a special case that  $k\leq n^{1-\delta}$ for some fixed constant  $\delta
\in(0, 1]$, we can further improve the running time to linear for \tailtree. Note that in practice,
this  is a reasonable assumption which generalizes the assumption $k\leq n^{1/2-\delta}$ in the previous
work~\cite{hegde2014nearly}. Our main approach is to show that we can safely ignore many nodes
at low levels.

We divide the tree into two \layers. Let $\eta = \ceil{2\log\log n}$.
The first \layer\ is from level $1$ to $\eta$ and the
second \layer\ is from level $(\eta +1)$ to $\log (n+1)$. For the second \layer, we still use
\textsf{FastRSMinPlus} algorithm to maintain an approximate tail array. The difference is that for the first \layer, we show that we only need to consider at most $O( n^{1-\delta}/\e)$ nodes. Recall that
$T_{ij}$ is the perfect binary subtree rooted at $N_{ij}$, and $u_{ij}=\sum_{N_{i'j'}\in T_{ij}}x_{i'j'}$ is the
total subtree weight of $N_{ij}$. Note that there are at most
$\bigl((n+1)/\log^2 n\bigr)$ nodes at level $\eta$. Let $u$ be the $\ceil{(1 +\e )n^{1-\delta}/\e }$-largest total subtree weight among these nodes $\{N_{\eta j}\}_j$. We argue that we can safely ignore all subtrees $T_{\eta j}$ if its corresponding total subtree weight $u_{\eta j}<u$.
The details can be found in Algorithm~\ref{alg:lisimtail}.

\begin{algorithm}[htb]
  \label{alg:lisimtail}
  \SetKwFunction{MinPlus}{MinPlus}
  \SetKwFunction{FastRSMinPlus}{FastRSMinPlus}
  \SetKwFunction{FindTree}{FindTree}
  \KwData{ A tree $T$ together with node weights $\{x_{ij}\}_{ij}$, an integer $k\in [n]$}
  \KwResult{A subtree $\hat{\Omega}$}
  Initialize  $\eta = \ceil{2\log\log n}$, $\e_i = \frac{\e}{3^{(i-\eta)/4}},  i \in [\eta, \log n] $ \;
  Compute  $u_{\eta j}=\sum_{N_{i'j'}\in T_{\eta j}}x_{i'j'}   ,  j \in [1, 2^{\log (n+1) -
    \eta}]$. Let $u$ be the $\ceil{\frac{(1+\e)n^{1-\delta}}{\e}}$-largest element among
  $\{u_{\eta j}\}_j$ (breaking ties arbitrarily).
  \textbf{Delete} all subtrees $T_{\eta j}$ from $T$ if $u_{\eta j}< u$ \;

  $ \hat{S}_{1j} \leftarrow   \{ s[0]=0, s[1] =x_{1j}\} $ , for each $ N_
  {1j}$ which is  not  deleted \;

  \For{$ i= 2 $ \emph{\KwTo} $ \eta $ }{
    \For{ each $ N_
  {ij}$ which is  not  deleted}{
      $ \hat{S}_{ij} \leftarrow$  \MinPlus{$\hat{S}_{i-1,2j-1}, \hat{S}_{i-1,2j}$} \;
      $\hat{s}[2^i-1]\leftarrow \sum_{N_{i'j'}\in T_{ij}}x_{i'j'}$. Let $\hat{S}_{ij}\leftarrow \hat{S}_{ij}\cup \{\hat{s}[2^i-1]\}$ \;
    }
  }

  For each $N_{\eta j}$ which is not deleted, $\hat{S}_{\eta j}\leftarrow$ an $\e_{\eta}$-\RS\ which approximates $\hat{S}_{\eta j}$ \;

  \For{$ i= \eta +1$ \emph{\KwTo} $ \log (n+1)$ }{
    \For{$ j = 1$ \emph{\KwTo} $2^{\log (n+1) - i}$}{
      $ \hat{S}_{ij} \leftarrow$  \FastRSMinPlus {$\e_{i-1}, \e_{i},
        \hat{S}_{i-1,2j-1}, \hat{S}_{i-1,2j}$}  \;
        $\hat{s}[2^i-1]\leftarrow \prod_{l=\eta}^{i} (1+\e_{l}) \cdot \sum_{N_{i'j'}\in T_{ij}}x_{i'j'}$.
        Let $\hat{S}_{ij}\leftarrow \hat{S}_{ij}\cup \{\hat{s}[2^i-1]\}$ \footnotemark \;
    }
  }
  Let $\hat{s}[L]\in \hat{S}_{\log (n+1),1}$ be the smallest element of index $L$ satisfying that $L\geq n-k$.

  \KwRet $\hat{\Omega}\leftarrow$ \FindTree{$L, T$}

  \caption{LinearTailTree: A linear time $(1+\e)$-approximation for \tailtree.}
\end{algorithm}

\begin{thm}
  \label{thm:lisimtail}
  Algorithm~\ref{alg:lisimtail} is a $(1+\e)$-approximation algorithm  with running time
  $O(n+ \e^{-2} n / \log n)$ for \tailtree\ if $k \leq n^{1-\delta}$.
\end{thm}

\begin{proof}

  We first prove the correctness.
  Let $\calC=\{T_{\eta j}: \ u_{\eta j}\geq u\} $ be the collection of those subtrees with total
  subtree weight at least $u$. Let $\bcalC=\{T_{\eta j}\}_j \setminus \calC$ be the complement of
  $\calC$.
  We argue that the influence caused by deleting the subtrees in $ \bar{\mathcal{C}}$ in Step 2 is
  negligible. Let $\Omega^*$ be the optimal support with the optimal tail value $\OPT =
  \sum_{N_i \notin \Omega^*} x_i$. Let $\check{\Omega}$ be the optimal subtree of the case, in which  we delete
  all subtrees in $\bar{\mathcal{C}}$.
  Let $\check{\OPT} = \sum_{N_i \notin \check{\Omega}}x_i$. Note that our algorithm obtains a
  $(1+\epsilon)$-approximation $\hat{\Omega}$ of $\check{\Omega}$ following from the analysis in
  Theorem~\ref{thm:simtail} and Corollary \ref{cor:fastmin+}.

  Thus, we only need to prove that $\check{\OPT} \leq (1+\epsilon) \OPT$. By the assumption that $k \leq n^{1-\delta}$, $\Omega^*$ contains at most $n^{1-\delta}$ nodes in
  $\bcalC$, which have a total weight at most $n^{1-\delta} u$. It means that $\check{\OPT}  - \OPT \leq  n^{1-\delta} u$. On the other hand, there are at
  least $n^{1-\delta}/\e$ subtrees in $\calC$ that do not intersect $\Omega^*$, since $\Omega^*$ can contain at most $n^{1-\delta}$ nodes in
  $\calC$. Thus, we have that $\OPT \geq n^{1-\delta}
  u/\e$. Hence, $\check{\OPT}  - \OPT \leq \e \OPT$ which proves the correctness.

  Then we analyze the running time. It costs $O(n)$ time to compute all $u_{\eta j}$ and $u$ in Step 2. For each node $N_{ij}$ at level $1\leq i\leq \eta$, it costs $O( 2^{2i-c\sqrt{i}})$ time to compute $\hat{S_{ij}}$ using the procedure \textsf{MinPlus}. Among each subtree in $\calC$, the number of nodes at level $1\leq i\leq \eta$
  is $2^{\eta-i} =O (\log^2 n / 2^i)$. On the other hand, there are at most $O(n^{1-\delta}/\e)$ trees in
  $\calC$. Thus, the total running time from Step 3 to Step 7 is
  $$
  O\left(\frac{n^{1-\delta}}{\e}\right) \sum_{i= 1}^{\eta}O \left(\frac{\log^2 n}{2^i}\cdot 2^{2i-c\sqrt{i}}\right) =
  O(\e^{-1} n^{1-\delta} \log^4 n) = o(n).
  $$

  Considering Step 8, it costs $O(2^{\eta})$ time for each node $N_{\eta j}$. Thus, the total running time for Step 8 is $O(\frac{n^{1-\delta}}{\e} \cdot 2^{\eta})=o(n)$. By Corollary~\ref{cor:fastmin+},  the construction time of all $\hat{S}_{ij}$ at level $\eta +1 \leq i \leq \log (n+1)$ from Step 9-12 is
  $$
\sum_{i =\eta+1}^{\log n}O \left(\frac{n+1}{2^i} \cdot \frac{\log n}{\e_i^2} \right) =O(\e^{-2} n/\log n )
=o(n).
$$
  Finally, the backtracking process \textsf{FindTree}$(L,T)$ in Step 14 costs $O(n)$ time . Therefore, the total running time of Algorithm \ref{alg:lisimtail} is $O(n+ \e^{-2}n/\log n)$.
\end{proof}

\subsection{A Linear Time Algorithm for \headtree }
\label{sec:headlinear}

Now we consider the \headtree\ version. Recall that our goal is to find a subtree $\Omega\in
\mathbb{T}_k$ such that $\sum_{N_{ij} \in \Omega}x_{ij} \geq (1-\e)\sum_{N_{ij} \in \Omega^*}x_{ij}$, where
$\Omega^*$ is the optimal solution of the \headtree\ problem.
In this subsection, we denote $\OPT_H=\sum_{N_i\in \Omega^*}x_i$ to be the optimal head value for \headtree.
Our framework is similar to the framework for \tailtree. We again assume that each node weight $x_{ij}$ is an integer among  $[0, O(n\log n / \epsilon)]$. Thus
there are at most $O(n\log n / \epsilon)$ different weight values. Similar to \tailtree, we can remove this assumption by a weight discretization technique, see Appendix \ref{sec:weight} for details. By this assumption, we still construct a dynamic program for \headtree. However, our techniques and definitions have some differences. We then show the differences in details in the following.

\topic{Approximate $(\max, +)$-Convolution} At first, we introduce
another concept called $(\max, +)$-convolution which is similar to $(\min, +)$-convolution (see
Definition~\ref{def:min+}).

\begin{defn}[$(\max, +)$-convolution]
  \label{def:max+}
  Given two  arrays $A = (a[0], a[1], a[2], \ldots, a[m_1])$ and $B = (b[0], b[1], b[2], \ldots,
  b[m_2])$, their  $(\max,   +)$-convolution  is the array $S = (s[0], s[1], s[2], \ldots ,$ $
  s[m_1+m_2])$ where $s[t] =  \max_{i=0}^{t} \{a[i] + b[t-i]\} , t \in  [0, m_1+m_2]$.
\end{defn}

The only difference from $(\min, +)$-convolution is that $s_t = \max_i \{a_i + b_{t-i}\}$. In fact,
these two definitions are equivalent. Suppose that  $-S = (-s[0], -s[1],
\ldots , -s[m_1+m_2])$ is the $(\min, +)$-convolution of $-A= (-a[0],
-a[1],   \ldots, $ $-a[m_1])$ and $-B= (-b[0], -b[1], \ldots, -b[m_2])$. Then $S = (s[0], s[1],
\ldots , s[m_1+m_2])$ is exactly the $(\max,+)$-convolution of  two arrays $A$ and $B$.

For each node $N_{ij}$ on the tree $T$, we define $S_{ij}=(s[0], s[1],\ldots, s[2^i-1])$ to be the
head array of $N_{ij}$, where each element $s[l]$ represents the optimal head value for \headtree\
on $T_{ij}$, i.e., $s[l]=\max_{\Omega\in \mathbb{T}_{l} (T_{ij})}\sum_{N_{i'j'}\in T_{ij}}
x_{i'j'}$. In fact, the array $S_{ij}$ can be achieved through computing the $(\max, +)$-convolution
from the arrays $S_{i-1, 2j-1}$ and $S_{i-1, 2j}$ of its two children.\footnote{Note that for each
  element $s[l]\in S_{ij}$ $(1\leq l\leq 2^i-1)$, we have that
  $s[l]=x_{ij}+\max_{t}\{a[t]+b[l-1-t]\}$ where $a[t]\in S_{i-1,2j-1}$ and $b[l-1-t]\in
  S_{i-1,2j}$.}
Similar to \tailtree, our key approach is to maintain a head sequence $\hat{S}_{ij}$ as an
approximation of $S_{ij}$ which reduces the running time. We first introduce some concepts to
describe $\hat{S}_{ij}$.

\begin{defn}[Head-Completion of $\alpha$-\RS]
  \label{def:headextended}
  Consider an $\alpha$-\RS\ $\hat{A} =(\hat{a}[i_1], \hat{a}[i_2] \ldots,  \hat{a}[i_m])$. Define its
  head-completion of cardinality $M$ by an array $ A' =(a'[0], a'[1], \ldots, a'[M]) $ satisfying
  that: 1) If $0\leq t \leq i_1-1, a'[t] = 0$; 2)  If $i_v\leq t \leq i_{v+1}-1 \ (1\leq v\leq m-1),
  a'[t] = \hat{a}[i_{v}]$; 3) If $i_m\leq t \leq M, a'[t] = \hat{a}[i_m]$.
\end{defn}

\begin{defn}[Head-Sequence Approximation]
\label{def:headseqapp}

  Given two $n$-length non-decreasing arrays $A' = (a'[0], a'[1], \ldots, a'[n])$ and $A = (a[0], a[1], \ldots,
  a[n])$, we say $A'$ is an $\alpha$-head-approximation of  $A$  if for any $i$, $(1-\alpha)a[i]
  \leq a'[i]  \leq a[i] $. We say an $\alpha$-\RS\ $\hat{A}$ head-approximates an array $A$ if its
  head-completion $A'$ of cardinality $n$ is an $\alpha$-head-approximation of $A$.

\end{defn}

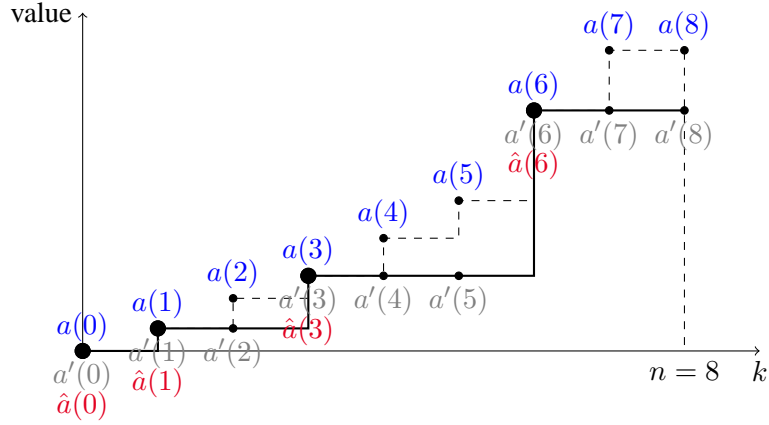
\begin{figure}[t]
  \centering
  \begin{tikzpicture}
    \draw[<->] (0,4.5) node[left]{value} -- (0,0) -- (9,0) node[below]{$k$} ;
    \draw[dashed] (0,0) -- (1,0) -- (1,0.3) -- (2,0.3) -- (2,0.7) --  (3,0.7) -- (3, 1) -- (4,1) --
    (4, 1.5) -- (5,1.5) -- (5, 2) -- (6, 2) -- (6, 3.2) --(7,3.2) --(7,4) --(8,4)  -- (8,0) ;
    \draw[thick] (0,0) --(1,0) -- (1,0.3 )-- (3 ,0.3) --(3, 1) -- (6 ,1) -- (6 , 3.2) -- (8,3.2);

    \draw[fill] (0,0) circle (0.1)  ;        \node[gray, below ] at (0,0) {$a'(0)$} ;
    \draw[fill] (1,0.3) circle (0.1)  ;     \node[gray, below ] at (1,0.3) {$a'(1)$} ;
    \draw[fill] (3,1) circle (0.1)  ;        \node[gray, below ] at (3,1) {$a'(3)$} ;
    \draw[fill] (6,3.2) circle (0.1) ;      \node[gray, below  ] at (6, 3.2) {$a'(6)$} ;

    \draw[fill] (0,0) circle (0.1)  ;        \node[red, below ] at (0,-0.4) {$\hat{a}(0)$} ;
    \draw[fill] (1,0.3) circle (0.1)  ;     \node[red, below ] at (1,-0.1) {$\hat{a}(1)$} ;
    \draw[fill] (3,1) circle (0.1)  ;        \node[red, below ] at (3,0.6) {$\hat{a}(3)$} ;
    \draw[fill] (6,3.2) circle (0.1) ;       \node[red, below] at (6,2.8) {$\hat{a}(6)$} ;

    \draw[fill] (2, 0.3) circle (0.05) ;    \node[ gray, below] at (2,0.3) {$a'(2)$} ;
    \draw[fill] (4,1) circle (0.05) ;        \node[gray, below ] at (4,1) {$a'(4)$} ;
    \draw[fill] (5,1) circle (0.05) ;        \node[gray, below ] at (5,1) {$a'(5)$} ;
    \draw[fill] (7,3.2) circle (0.05) ;     \node[gray, below  ] at (7,3.2) {$a'(7)$} ;
    \draw[fill] (8,3.2) circle (0.05) ;     \node[gray, below  ] at (8,3.2) {$a'(8)$} ;

    \draw[fill] (2, 0.7) circle (0.05) ;
    \draw[fill] (4, 1.5) circle (0.05) ;
    \draw[fill] (5, 2) circle (0.05) ;
    \draw[fill] (7, 4) circle (0.05) ;
    \draw[fill] (8, 4) circle (0.05) ;

    \node[blue, above ] at (0,0) {$a(0)$} ;
    \node[blue, above ] at (1 ,0.3) {$a(1)$} ;
    \node[blue, above ] at (2,0.7) {$a(2)$} ;
    \node[blue, above ] at (3,1) {$a(3)$} ;
    \node[blue, above ] at (4,1.5) {$a(4)$} ;
    \node[blue, above ] at (5, 2) {$a(5)$} ;
    \node[blue, above ] at (6,3.2) {$a(6)$} ;
    \node[blue, above ] at (7,4) {$a(7)$} ;
    \node[blue, above ] at (8,4) {$a(8)$} ;

    \node [below] at (8,0) {$n = 8$} ;

  \end{tikzpicture}

  \caption{ The figure illustrates the concepts $\alpha$-\RS\ and its head-completion. Here,
    $\hat{A} = (\hat{a}[0], \hat{a}[1],$ $ \hat{a}[3],\hat{a}[6])$ is an  $\alpha$-\RS.  The array
    $(a'[0]=\hat{a}[0], a'[1],\ldots, a'[8])$  is the head-completion of cardinality 8 of
    $\hat{A}$. By this figure, we can see that the $\alpha$-\RS\ $\hat{A}$ head-approximates the
    array $A =  (a[0], a[1], $ $\ldots, a[8])$. }

  \label{fig:max+}
\end{figure}

Figure~\ref{fig:max+} illustrates these definitions. Note that the above definitions have some
differences from in \tailtree.  By comparing Figure \ref{fig:min+} and \ref{fig:max+}, we can see
the differences. Now we are ready to define the concept of $(\alpha,\beta)$-\RS\
$(\max,+)$-convolution.

\begin{defn}[$(\alpha,\beta)$-\RS\ $(\max,+)$-convolution]
  \label{def:sparsemax} Given two $\alpha$-\RS s   $\hat{A} $  and $ \hat{B} $, suppose $A'$ and $B'$ are
  their head-completions of cardinality $M_1$ and $M_2$ respectively.  Suppose the array $ S$ is the $(\max,+)$-convolution of $A'$ and $B' $.  We call $\hat{S}$ an $(\alpha, \beta)$-\RS\ $(\max,
  +)$-convolution of $\hat{A}$ and $\hat{B}$ if $\hat{S}$ is a $\beta$-\RS\ which head-approximates the array $S$.
\end{defn}

Similar to Lemma~\ref{lm:fastmin+}, we have the following lemma.

\begin{lemma}
  \label{lm:fastmax+}
 Consider two $\alpha$-\RS s $\hat{A}=(\hat{a}[i_1], \hat{a}[i_2],\ldots, \hat{a}[i_{m_1}])$ where each $\hat{a}[i_w]$ ($1\leq w\leq m_1$) satisfies that either $\hat{a}[i_w]=0$ or $\hat{a}[i_w]\geq 1$, and $\hat{B}= (\hat{b}[j_1], \hat{b}[j_2],
  \ldots, \hat{b}[j_{m_2}])$ where each $\hat{b}[j_w]$ ($1\leq w\leq m_2$) satisfies that either $\hat{b}[j_w]=0$ or $\hat{b}[j_w]\geq 1$. Let $U = \max\{\hat{a}[i_{m_1}] , \hat{b}[j_{m_2}]\}$ and $M =\max\{ i_{m_1}, j_{m_2}\}$. Let $0\leq \beta\leq \alpha \leq 1$ be two constants. There exists an algorithm $\mathsf{FastRSMaxPlus}(\alpha, \beta,\hat{A},\hat{B})$
  computing an $(\alpha,\beta)$-\RS\ $(\max, +)$-convolution $\hat{S}$ of $\hat{A}$ and $\hat{B}$ in  $O (\min\{ \frac{\log_{1+\beta} U }{\alpha}, \frac{M}{\alpha}\})  $ time.
\end{lemma}

Using the same scheme as Algorithm \textsf{FastTailTree} in
Section~\ref{sec:nearlinear}, we can design a $(1- \epsilon)$-approximation algorithm
for \headtree\ with running time $ O( \e^{-1}n(\log\log\log n)^2 ) $. One difference is that we
compute $\hat{S}_{ij}$ by an approximate $(\max, +)$-convolution scheme \textsf{FastRSMaxPlus} by
Lemma \ref{lm:fastmax+}. The other difference is that after we compute the sequence $\hat{S}_{\log
  (n+1),1}$ for the root node, we find the largest element $\hat{s}[L]\in \hat{S}_{\log (n+1),1}$ of
index $L$ satisfying that $L\leq k$ and return a solution $\hat{\Omega}$ by a backtracking process.

\topic{A linear time algorithm for \headtree} In fact, we can improve the running time to linear by
some additional properties of \headtree. Let
$\xi = \lceil \log\log n - \log(1/\e) $ $ -\log\log\log n \rceil$, $\eta =  \ceil{\log\log n +\log(1/\e)}$. As
in Algorithm \textsf{FastTailTree}, the time cost of the second \layer\ (i.e., from level $\xi +1$
to level $\eta$) is the bottleneck.
 Fortunately for \headtree, we can speed up the
second \layer\ by a new weight discretization technique. Recall that $\e $ is a constant number.

Then we show how to compute an array $\hat{S}_{\eta j}$ as an $\e$-head-approximation of $S_{\eta
  j}$ for all nodes $N_{\eta j}$ in linear time. We first divide the
array $\hat{S}_{\eta j}$ into two sub-arrays. One sub-array consists of the first $\ceil{2\log\log
  n}$ elements $\hat{s}[l]$ $(0\leq l \leq \ceil{2\log\log n}-1$. The other sub-array
consists of the remaining elements $\hat{s}[l]$ $(\ceil{ 2\log\log n}\leq l\leq 2^{\eta}-1)$.  In the following, we show how to compute these two sub-arrays
respectively.

\topic{Case 1, $\hat{s}[l] \in \hat{S}_{\eta j}, 0\leq l \leq \ceil{2\log\log n}-1$}
We still compute  $\hat{S}_{\xi j}$ through the look-up table method as in Step 1 of Algorithm \textsf{FastTailTree}.
Then for any node $N_{ij}$ at level $\xi+1\leq i \leq \eta$, we construct a sub-array $\hat{S}_{ij}$ by computing an exact $(\max,+)$-convolution from its two children, while we only compute the first $\ceil{2\log\log n}$ elements $\hat{s}[l]\in \hat{S}_{ij}$ $(0\leq l\leq \ceil{2\log\log n}-1)$.

\topic{Case 2, $\hat{s}[l] \in \hat{S}_{\eta j}, \ceil{ 2\log\log n}\leq l\leq 2^{\eta}-1$} We consider a more careful weight discretization. Consider the perfect binary subtree $T_{\eta j}$ rooted at some node $N_{\eta j}$. Let $N_{\max}\in T_{\eta j}$ be the node of the largest weight $x_{\max}$.
Consider an element $s[l]\in S_{\eta j}$ ($ l \geq \ceil{2\log\log n} \geq \eta $) representing the optimal head value of sparsity $l$ for \headtree\ on $T_{\eta j}$. Then we have that $s[l] \geq x_{\max}$, since there exists a subtree rooted at $N_{\eta j}$ with
$l$ nodes and containing node $N_{\max}$.
\footnote{Note that this property is only satisfied in \headtree.}
We define the new discretized weight for each node
$N_{i'j'}\in T_{\eta j}$ to be $\hat{x}_{i'j'} = \floor{\frac{x_{i' j'} \log^2 n }{\e x_{\max}}}$.
After weight discretizing,  each node weight in $T_{\eta j}$ is an integer among the range $\left[0, \floor{\log^2 n / \e}\right]$.

Based on these node weights $\{\hat{x}_{i'j'}\}$, we again use the look-up table method to compute
all arrays  $\hat{S}_{\xi j}$ at level $\xi$. Similar to Step 2 of Algorithm \textsf{FastTailTree},
we compute an $(\e', \e')$-\RS\ $(\max, +)$-convolution from level $\xi+1$ to level $\eta$,
where $\e' = \e / (\eta -\xi)$. Now for a node $N_{\eta j}$, we obtain an approximate sequence and
we compute its head-completion $\hat{S}_{\eta j}$ of cardinality $2^{\eta}-1$. Finally for each
element $\hat{s}[l] \in \hat{S}_{\eta j}, \ceil{ 2\log\log n}\leq l\leq 2^{\eta}-1$, we multiply it
by a normalization factor $\e x_{\max} / \log^2 n$.

\vspace{0.2in}

Overall, we combine the above two sub-arrays, and obtain an approximate array $\hat{S}_{\eta j}$. We
will prove that $\hat{S}_{\eta j}$ is an $\e$-head-approximation of $S_{\eta j}$. Then we compute an
$\e$-\RS\ which head-approximates $\hat{S}_{\eta j}$ for each node $N_{\eta j}$. Finally we use the
similar technique as in Step 3 of Algorithm \textsf{FastTailTree}. For any node $N_{i j}$ at level
$\eta+1 \leq i\leq \log (n+1)$, we use Algorithm $\mathsf{FastRSMaxPlus}(\e_{i-1},\e_i,\hat{S}_{i-1,2j-1},\hat{S}_{i-1,2j})$
  to compute a sequence $\hat{S}_{ij}$, where $\e_i =  3^{-(i- \eta)/4}\e$ ($\eta \leq i\leq \log (n+1)$).

\eat{
\begin{algorithm}[t]
  \label{alg:linearhead}
  \SetKwFunction{FastRSMaxPlus}{FastRSMaxPlus}
  \SetKwFunction{FindTree}{FindTree}
  \KwData{ Tree-sparsity model $\mathcal{T}$}
  \KwResult{Subtree $\Omega$}
  Initialize: $ \xi= (\log\log n - \log (1/\e) - \log\log\log n), \eta= (\log\log n +\log (1/\e))$ ;
  $\e_i = \frac{\e}{3^{(i-\eta)/4}}, \quad  \e' = \frac{\e}{\eta - \xi} $ \;
  Process case 1. Get $s[l_r] \in \hat{S}_{\eta j}, l_r  < 2\log\log n$ \\
  Process case 2. Get $s[l_r] \in \hat{S}_{\eta j}, l_r  \geq 2\log\log n$ \\
  Combine case 1 and case 2 to get  $\hat{S}_{\eta j}$  \\
  Process each level $i$ of the first \layer\ by \FastRSMaxPlus{$\e_{i-1}, \e_{i}, N_{ij}.lcd, N_{ij}.rcd $} \\
  $\Omega \leftarrow $ \FindTree{$k, N_1$} \;
  \KwRet $\Omega$ \;

  \caption{LinearHeadTree: A linear time $(1-\e)$-approximation algorithm for \headtree }

\end{algorithm}
}

\begin{thm}
  \label{thm:linearhead}
  There is a  $(1-\e)$-approximation algorithm with running time $O(\e^{-1}n)$ for \headtree.
\end{thm}
\begin{proof}

  We first consider the running time. For Case 1, the running time for computing all $\hat{S}_{\eta j}$ is $o(n)$. For each node $N_{ij}$ at level $\xi+1\leq i \leq \eta$, since we only compute $\ceil{2\log\log n}$ elements, the running time for constructing the sub-array is $O(\log^2 \log n)$. Note that there are at most $O(n/2^{\xi})=O\bigl(n \log \log n/(\e \log n)\bigr)$ nodes. Thus, the total running time for Case 1 is $o(n)$. For Case 2, using the lookup table method costs $o(n)$ time. For each node $N_{ij}$ at level $\xi+1\leq i \leq \eta$, the running time for computing an $(\e', \e')$-\RS\ $(\min, +)$-convolution is
  $O(\log\log n / \e'^2)$.
  by Lemma~\ref{lm:fastmax+}.
  Thus, the total running time for this case is $\sum_{i \in (\xi, \eta]} 2^{-i}n \cdot  \log\log n / \e'^2 =o(n)$.
  For those nodes $N_{ij}$ at level $\eta+1\leq i \leq \log (n+1)$, by the same analysis in Theorem \ref{thm:simtail}, the total running time is $O(\e^{-1}n)$. Overall, the running time is $O(\e^{-1}n)$.

  Then we prove the approximation ratio. For Case 1, by Definition \ref{def:max+}, each element $\hat{s}[l]\in \hat{S}_{\xi j}$ $(0\leq l\leq \ceil{2\log\log n}-1)$ satisfies that $(1-\e)s[l]\leq \hat{s}[l]\leq s[l]$. For Case 2, we can show that the new weight discretization scheme leads to a $(1 - \e)$-approximation
  following from the same argument as in Lemma~\ref{lm:headrounding}. Then by the same argument as in Theorem \ref{thm:simtail}, we have that the processes $\mathsf{FastRSMaxPlus}$ from  level $\eta$ to $\log (n+1)$ compute an $(1-\e)$-head-approximation array for \headtree. Thus, the total approximation ratio is
  $(1-\e)$.
\end{proof}

Combining Theorem~\ref{thm:simtail}, ~\ref{thm:lisimtail} and \ref{thm:linearhead}, we obtain
Theorem~\ref{thm:linear}.
\eat{
\begin{remark}
  It is not different to extend our results to the $l_p$-norm for both \headtree\ and \tailtree. The only difference is that we compute the $l_p$-norm weight $|x_i|^p$ for each node $N_i\in T$ at the beginning. Then we run our algorithms for both \headtree\ and \tailtree\ using these $l_p$-norm weight $|x_i|^p$. We will obtain a $(1+\e)^{1/p}$-approximation for \headtree\ and a $(1-\e)^{1/p}$-approximation for \tailtree\ respectively. Thus, we only need to set the value of $\e$ to be $O(p\e)$ instead.

On the other hand, we can extend our algorithms to $b$-ary trees. Note that our algorithms are based on $(\min, +)$-convolutions
(or $(\max, +)$-convolutions). Consider any node $N$. We want to compute an approximate sequence $\hat{S}$. In a $b$-ary tree,  each node $N$ has $b$ children. Denote them by $N_1, N_2, \ldots, N_b$. We compute $\hat{S}$ by the following iterations.
\begin{flalign*}
  \begin{split}
    &\hat{S} \leftarrow \mathsf{MinPlus}(N_1, N_2) \\
    &\mathbf{For } \ i =3  \ \mathbf{to } \ d \  \mathbf{do: } \\
    &\qquad \hat{S} \leftarrow \mathsf{MinPlus}(\hat{S}, N_i) \\
    &\mathbf{Return } \quad \hat{S} \\
  \end{split}
\end{flalign*}

The iteration takes time $b$ times as much as before for a binary tree. Since we assume $b$ is a
constant integer, it does not affect the time complexity asymptotically.

\end{remark}

}

\subsection{Compressive Sensing Recovery}

By Theorem~\ref{thm:linear}, we can obtain a faster tree sparse recovery algorithm by the framework
AM-IHT in \cite{hegde2015approximation}. The framework AM-IHT is an iterative scheme. In each
iteration, we need to complete two matrix multiplications, a head-approximation, and a
tail-approximation projections.

\begin{reptheorem}{thm:tree} [Restated]
   Assume that $k \leq n^{1-\delta}$ ($\delta \in(0,1)$ is any fixed constant). Let $A\in
   \R^{m\times n}$ be a measurement matrix. Let $x\in \calM_{k}$ be an arbitrary signal in
  the tree sparsity model with dimension $n$, and let $y=Ax+e\in \R^m$ be a noisy measurement
  vector. Here $e\in \R^m$ is a noise vector. Then there exists an algorithm to recover a signal
  approximation $\hat{x}\in \calM_{k}$ satisfying $\|x-\hat{x}\|\leq C\|e\|_2$ for some constant $C$
  from $m=O(k)$ measurements. Moreover, the algorithm runs in $O( (n\log n+k^2\log n \log^2(k\log
  n))\log \frac{\|x\|_2}{\|e\|_2})$ time.
\end{reptheorem}

\begin{proof}
  Our theorem is very similar to Theorem 3 in \cite{hegde2014nearly}. The only difference is that we
  output a solution $\hat{x}\in \M_k$ instead of $\hat{x}\in \M_{ck}$ for some constant $c>1$. That is
  because the final solution is obtained from a tail oracle, and our tail oracle is a single-criterion
  oracle.
\end{proof}

\eat{

\begin{remark}
  Note that for head- and tail-approximations, we decrease the running time from $O(n\log n+ k\log^2
  n)$ to $O(n)$. Since a matrix multiplication can be highly parallelized. In fact, we make the
  running time faster for each iteration. On the other hand, the complexity of measurements is still
  $m=O(k)$. However, by the framework AM-IHT, our single-criterion algorithms can reduce the
  coefficient by a constant.
\end{remark}
}


\section{CEMD Model}
\label{sec:emd}

In this section, we discuss another structured sparsity model known as the \emph{Constrained EMD}
model~\cite{schmidt2013constrained}.

\subsection{A Single-Criterion Approximation Algorithm for Head-Approximation Projection}

We develop a single-criterion constant  approximation algorithm for the head approximation
projection in the CEMD model, improving the result in~\cite{hegde2015approximation} which relaxes the support space to
$\Omega\in \M_{k,B\log k}$.
We first use an EMD flow network~\cite{hegde2015approximation}, and similarly obtain two
supports $\Omega_l$ and $\Omega_r$. Then from these two supports, we  construct a
single-criterion constant factor approximate solution. Formally speaking, given
an arbitrary signal $x$, we want to find a support $\hat{\Omega}\in \M_{k,B}$ such that
$\sum_{x_{i,j} \in \hat{\Omega}} |x_{i,j}|^p \geq c\cdot \max_{\Omega \in \M_{k,B}} \sum_{x_{i,j} \in
  \Omega} |x_{i,j}|^p $ for some fixed
constant $c \in (0,1]$

\topic{Step 1: Constructing an EMD flow network}
We first recall the EMD flow network construction defined in~\cite{hegde2015approximation}. See
Figure~\ref{fig:emdflow} as an example.

\begin{defn} [EMD flow network]
  \label{def:emdflow}
  For a given signal $X$, sparsity $k$, and a parameter $\lambda>0$, the flow network
  $G_{X,k,\lambda}$ is defined as follows:
  \begin{enumerate}
  \item Each entry $x_{i,j}\in X$ corresponds to a node $v_{i,j}$ for $i\in [h],j\in
    [w]$. Additionally, add a source node $\mu$ and a sink node $\nu$.
  \item Add an edge from every $v_{i_1,j}$ to every $v_{i_2,j+1}$ for $i_1, i_2\in [h], j\in
    [w-1]$. Moreover, add an edge from the source to every $v_{i,1}$ and from every $v_{i,w}$ to the
    sink.
  \item The capacity on every edge and node (except source and sink) is 1.
  \item The cost of node $v_{i,j}$ is $-|x_{i,j}|^p$. The cost of an edge from $v_{i_1,j}$ to
    $v_{i_2,j+1}$ is $\lambda |i_1 - i_2|$. The cost
    of the source, the sink, and each edge incident to the source or sink is 0.
  \item Both the supply at the source and the demand at the sink are $s (=\frac{k}{w})$.
  \end{enumerate}
\end{defn}

\begin{figure}[t]
  \centering

\begin{minipage}[c]{0.3\textwidth}
\centering
  \begin{tikzpicture}
    \newcommand{\myunit}{1 cm}
    \tikzset{
      node style sp/.style={draw,circle,minimum size=\myunit},
      node style ge/.style={circle,minimum size=\myunit},
      arrow style mul/.style={draw,sloped,midway,fill=white},
      arrow style plus/.style={midway,sloped,fill=white},
    }

    \matrix (A) [matrix of math nodes,%
    nodes = {node style ge},%
    left delimiter  = (,%
    right delimiter = )] at ( 0 ,0)
    {%
      1  & 1 & 4  \\
      3  &  2  & 0   \\
      0  & 2 & 3   \\
    };
  \end{tikzpicture}
\end{minipage}%
\begin{minipage}[c]{0.7\textwidth}
\centering
   \begin{tikzpicture}

    \SetGraphUnit{3}
    \GraphInit[vstyle = Dijkstra]
    \SetVertexMath

    \def \a {2}
    \def \b {2}
    \Vertex[Math, L =  \mu , x=0*\a, y= 1*\b]{s}
    \Vertex[Math, L = $0$, x = 1*\a,y=0*\b]{10}
    \Vertex[Math, L = $-3$, x=1*\a, y=1*\b]{11}
    \Vertex[Math, L = $-1$, x= 1*\a,y=2*\b]{12}
    \Vertex[Math, L = $-2$, x=2*\a, y=0*\b]{20}
    \Vertex[Math, L=  $-2$, x= 2*\a, y=1*\b]{21}
    \Vertex[Math, L = $-1$, x= 2*\a,y=2*\b]{22}
    \Vertex[Math, L = $-3$, x= 3*\a, y=0*\b]{30}
    \Vertex[Math, L = $0$, x= 3*\a, y=1*\b]{31}
    \Vertex[Math, L = $-4$, x=3*\a, y=2*\b]{32}
    \Vertex[Math, L = \nu, x= 4*\a, y=1*\b]{t}

    \tikzset{EdgeStyle/.style} = {->}
    \Edge[style = {->}](s)(10)
    \Edge[style = {->}](s)(11)
    \Edge[style = {->}](s)(12)
    \Edge[style = {->}](30)(t)
    \Edge[style = {->}](31)(t)
    \Edge[style = {->}](32)(t)
    \Edge[color = red, style = {->}, label = $0$](10)(20)      \Edge[color = red, style = {->},label
    = $\lambda$](10)(21)    \Edge[color = red, style = {->},label = $2\lambda$](10)(22)
    \Edge[style = {->}](11)(20)      \Edge[style = {->}](11)(21)    \Edge[style = {->}](11)(22)
    \Edge[style = {->}](12)(20)      \Edge[style = {->}](12)(21)    \Edge[style = {->}](12)(22)
    \Edge[style = {->}](20)(30)      \Edge[style = {->}](20)(31)    \Edge[style = {->}](20)(32)
    \Edge[style = {->}](21)(30)      \Edge[style = {->}](21)(31)    \Edge[style = {->}](21)(32)
    \Edge[style = {->}](22)(30)      \Edge[style = {->}](22)(31)    \Edge[style = {->}](22)(32)

  \end{tikzpicture}

\end{minipage}
  \caption{EMD flow network. The left matrix is the signal $X$. The right figure is its
    corresponding EMD flow network $G_{X, k, \lambda}$. }
  \label{fig:emdflow}

\end{figure}
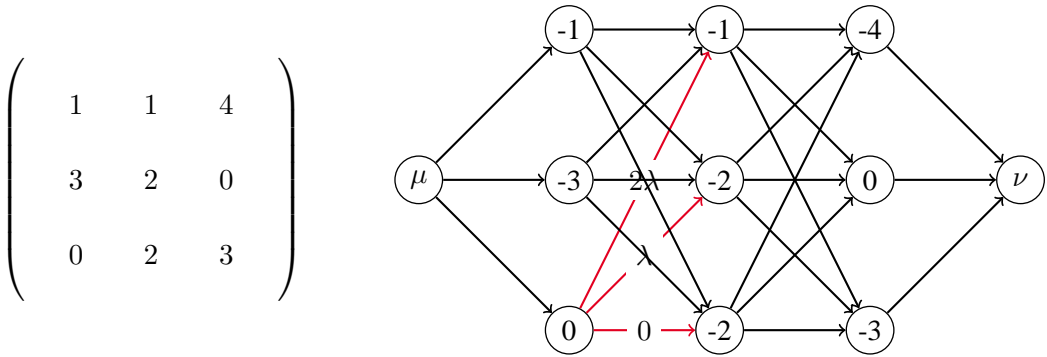

Since all edge capacities, the source supply, and the sink demand are integers, by Theorem 9.10
in~\cite{ravindra1993network}, we know that $G_{X,k,\lambda}$ always has an integer min-cost
max-flow. Note that this integer min-cost max-flow must be a set of disjoint paths through the
network $G_{X,k,\lambda}$, and it corresponds to a support in $X$. For a flow network
$G_{X,k,\lambda}$, we denote the support of this integer min-cost max-flow by
$\Omega_{\lambda}=\mathsf{MinCostFlow}(G_{X,k,\lambda})$. Thus, for any
$\lambda$, a solution of the min-cost max-flow problem on $G_{X,k,\lambda}$ reveals a subset $S$ of
nodes that corresponds to a support $\Omega_{\lambda}$ satisfying the following two properties: 1) in
each column, $\Omega_{\lambda}$ has exactly $s$ indices; 2) $\Omega_{\lambda}$ is the support which minimizes
$-\sum_{x_{i,j} \in \Omega} |x_{i,j}|^p+\lambda \EMD[\Omega]$ (also equivalent to maximize
$\sum_{x_{i,j} \in \Omega} |x_{i,j}|^p -\lambda \EMD[\Omega]$).

For convenience, we define $ \headvalue[\Omega]$ to be the head value $\sum_{x_{i,j} \in \Omega}
|x_{i,j}|^p$ of support $\Omega$ and denote the $\EMD[\Omega]$ by
$\emd[\Omega]$. In~\cite{hegde2015approximation}, we can obtain the following theorem by this flow
network.

\begin{thm} [Theorem 34 and 36 in~\cite{hegde2015approximation}]
  \label{thm:flownet}
  Let $\delta>0$, $x_{\min}=\min_{|X_{i,j}|>0} |X_{i,j}|^p$, and $x_{\max}=\max_{|X_{i,j}|>0}
  |X_{i,j}|^p$. There exists an algorithm running in $O(snh(\log\frac{n}{\delta})+\log
  \frac{x_{\max}}{x_{\min}})$ time, which returns two solutions
  $\Omega_l=\mathsf{MinCostFlow}(G_{X,k,l})$, and $\Omega_r=\mathsf{MinCostFlow}(G_{X,k,r})$. We have
  that $l,r\geq 0$, $l-r\leq \frac{\delta x_{\min}}{wh^2}$, and    $ \emd[\Omega_l]\leq B\leq
  \emd[\Omega_r]$.
\end{thm}

Then we show how to construct a single-criterion solution by $\Omega_l$ and $\Omega_r$.

\topic{Step 2: Constructing a single-criterion solution}
By Theorem~\ref{thm:flownet}, assume that we have two solutions $\Omega_l$ and $\Omega_r$ now. We
want to construct a single-criterion solution which is also a constant approximation. Note
that $\Omega_l\in \M_{k,B}$ and $\Omega_r$ may not be in $\M_{k,B}$. We first construct a
single-criterion solution $\Omega_r'$ based on $\Omega_r$ such that $\headvalue[\Omega_r'] \geq
\headvalue[\Omega_r] \cdot (2(\floor{\emd[\Omega_r]/B} +1))^{-1}$. We need the
following lemma for preparation.

\begin{lemma}
  \label{lm:pathdecom}
  Given any path $P$ on the flow network $G_{X,k,\lambda}$ from source to sink, let $\Omega_P$ be the
  support of $P$. Let $d\geq 1$ be  some positive integer.
There exists an $O(n)$ time algorithm which finds another path $P'$ with
  support $\Omega_{P'}$ satisfying that $\emd[\Omega_{P'}] \leq \emd[\Omega_{P}] / d$, and
  $\headvalue[\Omega_{P'}] \geq \headvalue[\Omega_{P}]/2d$.
\end{lemma}

\begin{proof}
  W.l.o.g., assume that the lowest node on path $P$ is at row 1.
  Consider the row $L_t$ which separates the lowest $t$ rows and the upper $h-t$ rows. Row $L_t$
  decomposes the path $P$ into two paths $\check{P}_t$ and $\hat{P}_t$. Specifically, for any edge
  $(v_{i_1, j}, v_{i_2, j+1})$, we add two edges in  $\check{P}_t$ and $\hat{P}_t$ respectively as
  follows.

  \begin{itemize}
  \item If $i_1 > t $ and $i_2 > t$, we add the edge $(v_{i_1, j}, v_{i_2, j+1})$ in $\hat{P}_t$ and add
    the edge $(v_{t, j}, v_{t, j+1})$ in $\check{P}_t$.  Similarly, if $i_1 \leq  t $ and $i_2 \leq
    t$, we add the edge $(v_{t+1, j}, v_{t+1, j+1})$ in
    $\hat{P}_t$ and add the edge $(v_{i_1, j}, v_{i_2, j+1}) $ in $\check{P}_t$.
  \item If $i_1 > t $ and $i_2 \leq  t$, we add the edge $(v_{i_1, j}, v_{t+1, j+1})$ in $\hat{P}_t$
    and add the edge $(v_{t, j}, v_{i_2, j+1}) $ in $\check{P}_t$. Similarly, if $i_1 \leq t $ and $i_2
    >  t$, we add the edge $(v_{t+1, j}, v_{i_2, j+1})$ in $\hat{P}_t$   and add the edge $(v_{i_1,
      j}, v_{t, j+1}) $ in $\check{P}_t$.
  \end{itemize}

See Figure~\ref{fig:path} as an example. Suppose for an edge  $(v_{i_1, j}, v_{i_2, j+1})$ in $P$,
we add an edge $(v_{\hat{i}_1, j},v_{\hat{i}_2, j+1}) $ in path $\hat{P}_t$ and an edge
$(v_{\check{i}_1, j},v_{\check{i}_2, j+1}) $. It is not difficult to check that $|i_2 - i_1| \geq
|\hat{i}_1 - \hat{i} _2 | + |\check{i}_1 - \check{i}_2 |$. Moreover, $\Omega_{\hat{P}_t} \cup
\Omega_{\check{P}_t} \supset \Omega_{P}$.  Thus,
  \begin{equation}
    \label{eq:emdpath}
    \headvalue[\Omega_{\check{P}_t}] + \headvalue[\Omega_{\hat{P}_t}] \geq \headvalue[\Omega_P], \;
    \emd[\Omega_{\check{P}_t}]+\emd[\Omega_{\hat{P}_t}]  \leq \emd[\Omega_{P}].
  \end{equation}

 Also observe that as $t$  increases, $\emd[\Omega_{\check{P}_t}]$ is non-decreasing.

  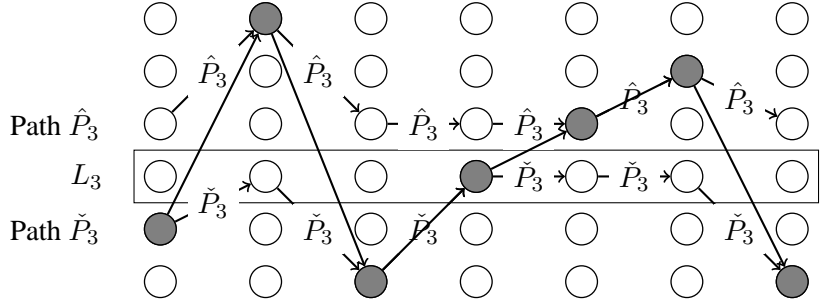
\begin{figure}[t]
    \centering
    \begin{tikzpicture}[scale = 0.7]
      \GraphInit[vstyle = Hasse]
      \tikzstyle{EdgeStype/.style}=[post]
      \foreach \i in {0,1,2,3, 4,5,6}
      {
        \foreach \j in {0,1,2,3,4,5}
        {
          \Vertex[x=2*\i, y=\j]{\i\j}
        }
      }

      \draw (-0.5,1.5) rectangle (12.5, 2.5);
      \node (l) at (-1.4,2) {$L_{3}$};
      \node (p1) at (-2,1) {Path $\check{P}_{3}$};
      \node (p2) at (-2,3) {Path $\hat{P}_{3}$};
      \Edge[style = {->},label = $\check{P}_3$ ,color = red](01)(12)
      \Edge[style = {->},label = $\check{P}_3$,color = red](12)(20)
      \Edge[style = {->},label = $\check{P}_3$,color = red](20)(32)
      \Edge[style = {->},label = $\check{P}_3$,color = red](32)(42)
      \Edge[style = {->},label = $\check{P}_3$,color = red](42)(52)
      \Edge[style = {->},label = $\check{P}_3$,color = red](52)(60)

      \Edge[style = {->},label = $\hat{P}_3$,color = blue](03)(15)
      \Edge[style = {->},label = $\hat{P}_3$,color = blue](15)(23)
      \Edge[style = {->},label = $\hat{P}_3$,color = blue](23)(33)
      \Edge[style = {->},label = $\hat{P}_3$,color = blue](33)(43)
      \Edge[style = {->},label = $\hat{P}_3$,color = blue](43)(54)
      \Edge[style = {->},label = $\hat{P}_3$,color = blue](54)(63)

      \AddVertexColor{gray}{01,15,20,32,43,54,60}
      \Edge[style = {->}](01)(15)
      \Edge[style = {->}](15)(20)
      \Edge[style = {->}](20)(32)
      \Edge[style = {->}](32)(43)
      \Edge[style = {->}](43)(54)
      \Edge[style = {->}](54)(60)
    \end{tikzpicture}
    \caption{Path decomposition. The original path (the gray one) can be divided into two parts (the
      blue and red ones respectively) }
    \label{fig:path}
  \end{figure}

  If $ \emd[\Omega_{P}]$, then the path $P$ itself satisfies the lemma. Thus, we assume that $
  \emd[\Omega_{P}]>0$. We then
  prove the lemma by induction on $d$. If $d=1$, the path $P$ itself satisfies the lemma. Suppose the
  lemma is true for any positive integer no more than $d-1$. Now we consider the integer $d$.

  We first find the highest row $L_t$ such that $  \emd[\Omega_{\check{P}_t}] \leq  \emd[\Omega_{P}]/d$, and
  $ \emd[\Omega_{\check{P}_{t+1}}]>  \emd[\Omega_{P}]/d $. Note that
  such an index $t$ must exist, since
  $$  \emd[\Omega_{\check{P}_0}]  =0 \leq
  \emd[\Omega_{P}]/d <    \emd[\Omega_{P}] = \emd[\Omega_{\check{P}_h}]. $$
  Note  that  row $L_{t}$ is
  a path with $  \emd[\Omega_{l_t}] = 0$ where  $\Omega_{l_t}$ is the support of $L_t$. We
  distinguish three cases.
  \begin{enumerate}
  \item If $  \headvalue[\Omega_{\check{P}_t}] \geq  \headvalue[\Omega_{P}] /2d$, then the
    path $P_t$ satisfies the lemma.
  \item If $ \headvalue[\Omega_{\check{P}_t}] < \headvalue[\Omega_{P}] /2d $  and
    $\headvalue[{\Omega_{L_{t+1}}}] \geq \headvalue[\Omega_{P}] /2d$, then the path $L_{t+1}$
    satisfies the lemma.
  \item If $\headvalue[\Omega_{\check{P}_t}] < \headvalue[\Omega_{P}] /2d $  and
    $\headvalue[{\Omega_{L_{t+1}}}] < \headvalue[\Omega_{P}] /2d $, then we have that $
    \headvalue[{\Omega_{\check{P}_{t+1}}}] <  \headvalue[\Omega_{P}] /d$ and
    $ \emd[{\Omega_{\check{P}_{t+1}}}] >  \emd[\Omega_{P}]/d$. Thus, according to
    Inequalities~\ref{eq:emdpath},  we have that
  $$\headvalue[\Omega_{\hat{P}_{t+1}}] \geq  \headvalue[\Omega_{P}] -
    \headvalue[{\Omega_{\check{P}_{t+1}}}]  \geq (1 -1/d) \headvalue[\Omega_{P}] ,$$
    and
    $$\emd[\Omega_{\hat{P}_{t+1}}] \leq   \emd[\Omega_{P}] -
    \emd[{\Omega_{\check{P}_{t+1}}}] \leq (1-1/d) \emd[\Omega_{P}].$$
    By  induction, we can find a path $P'$ from path $\hat{P}_{t+1}$, such that
    $$
    \headvalue[\Omega_{P'}]\geq \headvalue[\Omega_{\hat{P}_{t+1}}] / 2(d-1) >
    \headvalue[\Omega_P]/2d, \; \emd[\Omega_{P'}] \leq \emd[\Omega_{\hat{P}_{t+1}}]/(d-1)<
    \emd[\Omega_{P}]/d.
    $$
  \end{enumerate}

  By the above discussion, we prove the lemma.
\end{proof}

Note that $\Omega_r$ consists of $s$ disjoint paths. According to Lemma~\ref{lm:pathdecom}, we can
construct a single-criterion solution $\Omega_r'$ as follows.
\begin{corollary}
  \label{cor:pathdecom}
  Let $ d = \floor{ \emd[\Omega_r] / B}$.  We can construct a support $\Omega_r'\in \calM_{k,B}$ such that $
  \headvalue[\Omega_r'] \geq \headvalue[\Omega_r] / 2(d+1)$ in $O(ns)$ time.
\end{corollary}

 We next compare $ \headvalue[\Omega_r'] $
with $\headvalue[\Omega_l]$. If $ \headvalue[\Omega_r'] > \headvalue[\Omega_l] $, then we output
$\Omega_r'$ as our solution. Otherwise, we output $\Omega_l$
as our solution. By the following lemma, we show that our solution is a constant approximation.

\begin{lemma}
  \label{lm:emdcon}
  Suppose $\OPT=\max_{\Omega \in \M_{k,B}} \headvalue[\Omega]$.  Then we have that
  $$
  \max\{ \headvalue[\Omega_r'], \headvalue[\Omega_l]\} \geq (\frac{1}{4}-\delta) \OPT.
  $$
\end{lemma}

\begin{proof}

  Recall that  $\mathsf{MinCostFlow}$ solves the min-cost max-flow $\Omega_{\lambda}$ of the graph
  $G_{X,k,\lambda}$, i.e.,
  $$
  \headvalue[\Omega_{\lambda}] - \lambda \cdot \emd[\Omega_{\lambda}] = \max\nolimits_{\Omega \in
    \M_{k,B}} \{ \headvalue[\Omega]- \lambda \cdot   \emd[\Omega] \}
  $$
   The value of objective is no less than $0$ for any $\lambda$ because  there exists some support
   $\Omega$ such that $ \emd[\Omega] = 0$ and $ \headvalue[\Omega] \geq 0 $ for any $\Omega$.

   We get $ \Omega_l$ and $\Omega_r$ from $\mathsf{MinCostFlow}$ algorithm for $\lambda$ equaling to
   $l$ and $r$ respectively.    Thus, we have $ \headvalue[\Omega_r] - r \cdot \emd[\Omega_r] \geq 0
   $.  Moreover, $ \headvalue[\Omega_r] \geq \OPT$. Suppose $\headvalue[\Omega^{*}] = \OPT$.
   If $ \headvalue[\Omega_r] < \OPT $, changing $\Omega_r$  to $\Omega^*$
   would increase the objective $ \headvalue[\Omega_r] -r \cdot \emd[\Omega_r] $ since $ \emd[\Omega_r]
   \geq B$, which yields a contradiction.

   Assume  that $ \emd[\Omega_r] \in [dB, (d+1)B ) $ for   some positive integer $d\geq 1$. We
   distinguish three cases.

  \begin{enumerate}

  \item If $d=1$, by Corollary~\ref{cor:pathdecom}, $\Omega_r'$ satisfies  $  \headvalue[\Omega_r'] \geq
    \headvalue[\Omega_r] / 2 \geq \OPT/2$ and $\emd[\Omega_r'] \leq \emd[\Omega_r] / 2 \leq B$.

  \item If $d\geq 2$ and $ \headvalue[\Omega_r] \geq  3d \OPT/4$, by
    Corollary~\ref{cor:pathdecom},  $\Omega_r'$ satisfies $  \headvalue[\Omega_r'] \geq
    \headvalue[\Omega_r] /2(d+1) \geq  3d\OPT/ 8(d+1)  \geq \OPT/4$ and
    $ \emd[ \Omega_r'] \leq \emd[\Omega_r]/(d+1) \leq B$.

  \item If $d\geq 2$ and $ \headvalue[\Omega_r] <  3d  \OPT/4 $,  we have that $\OPT >
    4rB/3$ since $  \headvalue[\Omega_r] \geq r \cdot \emd[\Omega_r] \geq rdB$. Then we have
    the following inequalities.
    \begin{align*}
      \headvalue[\Omega_l] -l \cdot \emd[\Omega_l]  &\geq \OPT- l \cdot B \\
      \headvalue[\Omega_l] & \geq \OPT- l \cdot B
            \geq \OPT-(r+l-r)B
            \geq \OPT-\frac{3\OPT}{4}-(l-r) B \\
                     & = \frac{1}{4}\OPT-\frac{x_{\min}\delta B}{wh^2}
                       \geq \frac{1}{4}\OPT-\delta x_{\min}
                       \geq (\frac{1}{4}-\delta)\OPT.
    \end{align*}

  Here,  the first inequality follows from the fact
  $$
      \headvalue[\Omega_l] -l \cdot \emd[\Omega_l]  = \max_{\Omega \in \M_{k,B}} \{ \|x_{\Omega }\|_p^p- l \cdot
  \EMD(\Omega )\}  \geq \OPT - l\cdot B.
  $$

  Besides, $l-r \leq  \frac{x_{\min}\delta B}{wh^2}$ follows from Theorem~\ref{thm:flownet}.

  \end{enumerate}
Overall, we prove the lemma.
\end{proof}

Combining Theorem~\ref{thm:flownet} and Lemma~\ref{lm:emdcon}, we have the following theorem.

\begin{reptheorem}{thm:emdcon}[Restated]
  Consider the CEMD model $\calM_{k, B}$ with $s=k/w$ sparse for each column and support-EMD $B$.
  Let $\delta \in (0,1/4)$, $x_{\min}=\min_{|X_{i,j}|>0} |X_{i,j}|^p$, and $x_{\max}=\max_{|X_{i,j}|>0}
  |X_{i,j}|^p$. Let $c=1/4 -\delta$. There exists an algorithm running in
  $O(shn \log\frac{n}{\delta}+\log \frac{x_{\max}}{x_{\min}})$ time, which returns a single-criterion
  $c^{1/p}$ approximation for head-approximation projection.
\end{reptheorem}

Note that the exponent $1/p$ of $c$ comes from $l_p$-norm.  

\subsection{Compressive Sensing Recovery}

Similar to tree sparsity model, our head oracle in Theorem~\ref{thm:flownet} can also lead to a
model-based compressive sensing recovery algorithm, combining with AM-IHT and the tail oracle in
\cite{hegde2015approximation}. We summarize our result as follows.

\begin{reptheorem}{thm:CEMD} [Restated]
  Let $A\in \R^{m\times n}$ be a measurement matrix. Let $x\in \M_{k,B}$ be an arbitrary signal in
  the CEMD model with dimension $n=wh$, and let $y=Ax+e\in \R^m$ be a noisy measurement vector. Here
  $e\in \R^m$ is a noise vector. Then there exists an algorithm to recover a signal approximation
  $\hat{x}\in \M_{k, 2B}$ satisfying $\|x-\hat{x}\|\leq C\|e\|_2$ for some constant $C$ from
  $m=O(k\log (B/k))$ measurements. Moreover, the algorithm runs in $O(n\log
  \frac{\|x\|_2}{\|e\|_2}(k\log n+\frac{kh}{w}(\log n+\log \frac{x_{\max}}{x_{\min}})))$ time, where
  $x_{\max}=\max |x_i|$ and $x_{\min}=\min_{|x_i|>0} |x_i|$.
\end{reptheorem}

\begin{proof}
  Our theorem is very similar to Theorem 37 in \cite{hegde2015approximation} except two
  improvements. The first improvement is that we reduce the number of measurements. That is because
  the head oracle in \cite{hegde2015approximation} outputs a solution in $\M_{k,\gamma B}$, and the
  number $m$ of measurements has the following bound:
  $$
  m=O(k\log\frac{\gamma B}{k}).
  $$
  In \cite{hegde2015approximation}, $\gamma=O(\log (k/w))$. In contrast, our head oracle confirms
  that $\gamma=1$ by Theorem~\ref{thm:flownet}.

  The second improvement is the running time. There are $O(\log \frac{\|x\|_2}{\|e\|_2})$ iterations
  in the framework AM-IHT. In each iteration, we need to complete two matrix multiplications, a
  head-approximation, and a tail-approximation. In \cite{hegde2015approximation}, the time
  complexity of a head oracle is $O(\frac{nkhB}{w})$, while the time complexity of our head oracle
  is exactly the same as the tail oracle in \cite{hegde2015approximation} by
  Theorem~\ref{thm:flownet}.
\end{proof}



\newpage
\bibliographystyle{acm}
\bibliography{Citation}

\appendix

\section{Some missing details}
\label{sec:extend}

 \topic{Extend the Algorithm for $l_1$-norm to $l_p$-norm} It is not different to extend our
 results to the $l_p$-norm for both \headtree\ and \tailtree. The
  only difference is that we compute the $l_p$-norm weight $|x_i|^p$ for each node $N_i\in T$ at the
  beginning. Then we run our algorithms for both \headtree\ and \tailtree\ using these $l_p$-norm
  weight $|x_i|^p$. We will obtain a $(1+\e)^{1/p}$-approximation for \headtree\ and a
  $(1-\e)^{1/p}$-approximation for \tailtree\ respectively. Thus, we only need to set the value of
  $\e$ to be $O(p\e)$ instead.

\topic{Extend the Algorithm for Binary Tree to $b$-ary} On the other hand, we can extend our
algorithms to $b$-ary trees. Note that our algorithms are based
on $(\min, +)$-convolutions
(or $(\max, +)$-convolutions). Consider any node $N$. We want to compute an approximate sequence
$\hat{S}$. In a $b$-ary tree,  each node $N$ has $b$ children. Denote them by $N_1, N_2, \ldots,
N_b$. We compute $\hat{S}$ by the following iterations.
\begin{flalign*}
  \begin{split}
    &\hat{S} \leftarrow \mathsf{MinPlus}(N_1, N_2) \\
    &\mathbf{For } \ i =3  \ \mathbf{to } \ d \  \mathbf{do: } \\
    &\qquad \hat{S} \leftarrow \mathsf{MinPlus}(\hat{S}, N_i) \\
    &\mathbf{Return } \quad \hat{S} \\
  \end{split}
\end{flalign*}

The iteration takes time $b$ times as much as before for a binary tree. Since we assume $b$ is a
constant integer, it does not affect the time complexity asymptotically.

\topic{FindTree Algorithm}
We give the $\mathsf{FindTree(L,T)}$ as follows.    By the backtracking process $\mathsf{FindTree}(L, T)$, we obtain a support $\hat{\Omega}$ with a
  tail value at most $\hat{s}[L]$ since each element $\hat{s}[L]$ is at least as large as the exact
  tail value $s[L]$ by the algorithm. On the other hand, $|\hat{\Omega}|\leq k$ since $L\geq n-k$.

Then, we analyse the running time of the backtracking
  process $\mathsf{FindTree}(L, T)$. In fact, for each element $\hat{s}\in \hat{S}$ of index $L$, we
  can save the two indices $L_1$ and $L_2$ satisfying the condition in Line 3 of Algorithm
  \ref{alg:findtree}, during constructing $\hat{S}$ in Algorithm
  \textsf{FastTailTree}.
  Thus, we only cost $O(1)$ time for each node in $\mathsf{FindTree}(L, T)$. Then the running time
  of $\mathsf{FindTree}(L, T)$ is $O(n)$.

\begin{algorithm}[h]
  \caption{$\mathsf{FindTree}(L, T)$}
\label{alg:findtree}
  \SetKwFunction{FindTree}{FindTree}

Suppose the root node of $T$ is $N$ and it maintains a sequence $\hat{S}$ computed by \textsf{FastTailTree}. Let $\hat{s}\in \hat{S}$ be the element of index $L$. Suppose $N_1$ and $N_2$ are
$N$'s two children \;

Suppose that $T_1$ and $T_2$ are the two subtrees rooted at $N_1$ and $N_2$ respectively. Suppose
that $\hat{S}_1$ and $\hat{S}_2$ are sequences maintained in $N_1$ and $N_2$ respectively, computed
by \textsf{FastTailTree} \;

If $L=|T|$, \FindTree{$L,T$} $\leftarrow \emptyset$. Otherwise, find indexes $L_1$ and $L_2$ satisfying
that: 1) $\hat{s}_1\in \hat{S}_{1}$ is of index $L_1$ and $\hat{s}_2\in \hat{S}_{2}$ is of index
$L_2$, 2) $L_1+L_2=L$, 3) $\hat{s}_1+\hat{s}_2=\hat{s}$ \;

  \FindTree{$L,T$} $\leftarrow $\FindTree{$L_1,T_{1}$}$\cup $ \FindTree{$L_2,T_{2}$} $\cup \{N\}$.
\end{algorithm}

\section{Weight Discretization in the Tree Sparse Model}
\label{sec:weight}

\topic{Weight Discretization for \tailtree} We first introduce a linear time $O(\log
n)$-approximation algorithm for \tailtree, which offers a criterion to discretize the weight.

\begin{enumerate}
\item For each node $N_{ij}$, denote the largest subtree
  rooted at $N_{ij}$ by $T_{ij}$. Compute the subtree weight $u_{ij} = \sum_{i,j: N_{ij} \in
    T_{ij}} x_{ij}$ of all nodes in the subtree $T_{ij}$. Let $u$ be the $k$th largest weight
  among $\{u_{ij} \}_{i,j }$.
\item Add all nodes with $u_{ij} >  u$ into $\Omega$ directly. Then do a BFS (breath-first-search) on tree $T$ and
  add all nodes with $u_{ij} = u$ into $\Omega$ until  $|\Omega|=k$. Denote
$ \sum_{N_{ij} \notin \Omega} x_{ij} $  by $\logapp$ and  return $\logapp$.
\end{enumerate}

We have $W\leq \log n \cdot \OPT$, which means that $W$ is a $\log n$-approximation for \tailtree.

\begin{lemma}
  \label{lm:taillog} The above algorithm is a $\log n$-approximation algorithm with running time
  $O(n)$ for \tailtree.
\end{lemma}

\begin{proof}
Observe that for any two nodes $N_{ij}$ and $N_{i'j'}$, if $N_{ij}$ is the ancestor of
$N_{i'j'}$, we have $u_{ij}\geq u_{i'j'}$. Combining this fact and the BFS procedure, we have
that the support $\Omega$ is a subtree
rooted at $N_{\log(n+1), 1}$. Then we analyze the time complexity and approximation ratio.

  The weight $u_{ij}$ is the summation of the weights
  of its  left subtree, right subtree and itself. We compute $u_{ij}$ from leafs to root. Hence, it takes $O(1)$ time
  to compute each $u_{ij}$. Constructing $\Omega$ needs $O(n)$ time since we only do a BFS. Thus, the
  total running time is $O(n)$.

  Finally we prove the approximation ratio. Recall that $\Omega^*$ is the optimal subtree rooted at
  $N_{\log{n+1}, 1}$. We have the following inequality.
  $$
  \logapp = \sum_{N_{ij} \notin \Omega} x_{ij} \leq \sum_{N_{ij} \notin
    \Omega} u_{ij} \leq \sum_{N_{ij} \notin \Omega^*} u_{ij}.
$$
  The last inequality follows from the fact that the algorithm selects the $k$ nodes with the
  largest weight $u_{ij}$. Note that each node appears in at most $\log n $ different subtrees
  $T_{ij}$ except the root node, we have
  $$
  \logapp \leq \sum_{N_{ij} \notin \Omega^*} u_{ij}  = \sum_{N_{ij} \notin \Omega^*} \sum_{N_{i'j'} \in
    T_{ij}} x_{i'j'} \leq \sum_{N_{ij} \notin \Omega^*} \log n \cdot  x_{ij} \leq \log n \cdot \OPT.
  $$

\end{proof}

Next, we show how to discretize the weights. For each node $N_{ij}$, if its node weight $x_{ij}\in [0,
W]$, we define $\hat{x}_{ij} = \ceil{\frac{ x_{ij} n \log n }{\epsilon \logapp}}$ to be the discretized
weight. Otherwise if $x_{ij} >W$, we define $\hat{x}_{ij}=\ceil{\frac{ n\log n }{\epsilon}}+n$.  By this
discretization, we have that each $\hat{x}_{ij}$ is an integer among  $[0, \frac{n \log
  n}{\e}+n]$. Assume that $\hat{\Omega}$ is the optimal solution for \tailtree\ based on the discretized
weights $\{\hat{x}_{ij}\}_{ij}$, together with a tail value $\widehat{\OPT}=\sum_{N_{ij} \notin
  \hat{\Omega}}\hat{x}_{ij}$. Denote $\OPT' =
\sum_{N_{ij} \notin \hat{\Omega}} x_{ij} $ to be the tail value of $\hat{\Omega}$ based on the original
weights $\{x_{ij} \}_{i,j }$. We prove the following lemma.

\begin{lemma}
  \label{lm:d1}
  $  \OPT' \leq  (1+\e)\OPT$.

\end{lemma}

\begin{proof}
Since
$\logapp \geq \OPT$, those nodes of weight larger than $\logapp$
must be in the optimal solution $\Omega^*$. Hence, we have that
$$\sum_{N_{ij} \notin \Omega^*} \hat{x}_{ij} =\sum_{N_{ij} \notin \Omega^*} \bigceil{ \frac{ x_{ij} n \log n }{\epsilon
    \logapp}}\leq \sum_{N_{ij} \notin \Omega^*} \frac{ x_{ij} n \log n }{\epsilon \logapp}+1= \frac{ \OPT
  \cdot n \log n }{\epsilon \logapp}+n-k < \bigceil{\frac{n \log n}{\e}}+n.
$$
By the construction of $\hat{\Omega}$, we have that
$$\widehat{\OPT}=\sum_{N_{ij} \notin \hat{\Omega}}\hat{x}_{ij} \leq \sum_{N_{ij} \notin \Omega^*} \hat{x}_{ij} <
\bigceil{\frac{n \log n}{\e}}+n.$$
By the above inequality, we conclude that all nodes of weight larger
than $ \logapp$ are also in the solution $\hat{\Omega}$. Thus, for any node $N_{ij} \notin \Omega^*
\cup \hat{\Omega}$, we have that $x_{ij} \leq \logapp$ and $\hat{x}_{ij} = \ceil{\frac{ x_{ij} n \log n }{\epsilon  \logapp}} \leq  \ceil{\frac{n \log
  n}{\e}} $ . By this observation, we have the following inequality.
  \begin{align*}
    \OPT &= \sum_{N_{ij} \notin \Omega^*} x_{ij}  \geq  \frac{\e  \logapp }{n \log n} \sum_{ N_{ij} \notin
           \Omega^*} (\hat{x}_{ij} -1)  \geq \frac{\e  \logapp }{n \log n} \sum_{ N_{ij} \notin \Omega^*} \hat{x}_{ij} -
           \frac{\e  \logapp }{n \log n} \cdot n \\
         & \geq   \frac{\e  \logapp }{n \log n} \sum_{ N_{ij} \notin \hat{\Omega}} \hat{x}_{ij} - \frac{\e  \logapp}{\log n}
           \geq  \frac{\e \logapp }{n \log n} \sum_{ N_{ij} \notin \hat{\Omega}} \bigceil{\frac{
           x_{ij} n \log n }{\epsilon  \logapp}} - \e \cdot \OPT  \\
         &   \geq \sum_{N_{ij} \notin \hat{\Omega}} x_{ij}  - \e \cdot \OPT \geq \OPT' - \e \cdot \OPT
  \end{align*}
  Here the fourth inequality follows from the fact that $\frac{\logapp}{\log n} \leq \OPT$ by
  Lemma~\ref{lm:taillog}.
\end{proof}

By Lemma~\ref{lm:d1}, we know that the influence caused by the weight discretization is
negligible. Note that all nodes of weight larger than $W$ are in the solution $\hat{\Omega}$.  W.l.o.g., we assume
that each node is of weight $x_{ij}\leq \logapp$ and $\hat{x}_{ij} \leq \ceil{\frac{n \log n}{\e}}$. From
now on, we focus on the discretized weight $\{\hat{x}_{ij} \}_{ij}$. For convenience, we use $x_{ij}$
to represent $\hat{x}_{ij}$.

\topic{Weight Discretization for \headtree} In order to discretize the weight, we still need to introduce
a linear time $O(1/\log n)$-approximation algorithm for \headtree\ as follows.

\begin{enumerate}
\item Let $Q$ be the collection of $\floor{k / \log n}$ nodes with the largest node weights (breaking ties arbitrarily).
\item For each node $N_{ij} \in Q$, append to the solution $\hat{\Omega}$ all nodes on the path from $N_{ij}$ to the root node. Let $\logapp = \sum_{N_{ij} \in \hat{\Omega}}x_{ij} $.
\end{enumerate}

Then we prove the head value $W\geq \OPT_H/3\log n$.

\begin{lemma}
\label{lm:headlog}
  $\hat{\Omega}$ is an $(1/3\log n)$-approximation for \headtree\ with running time $O(n)$.
\end{lemma}

\begin{proof}
  Note that the number of nodes in $\hat{\Omega}$ is at most $\log n \cdot \floor{ k / \log n } \leq k$. Thus $\hat{\Omega}$ is a feasible
  solution. On the other hand, assume that the minimum weight of nodes in $Q$ is $w$. Then the head
  value $\logapp$ is at least $W=\sum_{N_{ij} \in \hat{\Omega}} x_{ij} \geq w\cdot
  \floor{ k / \log n }$, while the optimal  solution $\OPT_H$ is at most $\sum_{N_{ij} \in \Omega^*} x_{ij}
  < \sum_{N_{ij} \in Q}x_{ij} +  k \cdot w\leq W+k\cdot w$. So we can conclude that $ \logapp \leq  \OPT_H <
  (2\log n+1) \logapp$.

  Consider the running time. We cost $O(n)$ time to construct the collection $Q$, and cost
  $O(|\hat{\Omega}|)=O(k)$ time to construct $\hat{\Omega}$. Overall, the running time is $O(n)$.
\end{proof}

Next, we show how to discretize the node weights by the criterion $W$. We define $\hat{x}_i$ to be
$\floor{\frac{kx_i}{\epsilon \logapp }}$. Assume the optimal solution for \headtree\ based
on the discretized node weights is $\hat{\Omega}$. Denote $\OPT_H' =\sum_{N_{ij}
  \in \hat{\Omega}} x_{ij}$. We next analyze the difference between two solutions $\Omega^*$ and
$\hat{\Omega}$. We have the following lemma.

\begin{lemma}
  \label{lm:headrounding}
  $ \OPT_H' \geq (1-\epsilon)  \OPT_H $.
\end{lemma}

\begin{proof}
  By the definition of $\OPT_H'$, we have that
  \begin{align*}
    & \OPT_H' = \sum_{N_{ij} \in \hat{\Omega}} x_{ij} \geq \frac{\e  \logapp}{k} \cdot
    \sum_{N_{ij} \in \hat{\Omega}}
      \floor{\frac{kx_{ij}}{\epsilon  \logapp}} \geq \frac{\e  \logapp}{k} \cdot \sum_{N_{ij}\in \Omega^*}
      \floor{\frac{kx_{ij}}{\epsilon  \logapp}} \\
    \geq & \frac{\e  \logapp}{k} \cdot \sum_{N_{ij} \in \Omega^*}
           \left(\frac{kx_{ij}}{\epsilon  \logapp}-1\right)=\OPT_H-\e \logapp \geq (1-\e)\OPT_H.
  \end{align*}
  The second inequality follows from the definition of $\hat{\Omega}$, and the last inequality follows from the fact that $ \logapp \leq \OPT_H $.

\end{proof}

By Lemma \ref{lm:headrounding}, we know that the loss caused by the weight discretization is
negligible. From now on, we focus on the discretized weights $\{\hat{x}_{ij}\}_{ij}$. For
convenience, we use $x_{ij}$ to represent $\hat{x}_{ij}$.
 In the following, we only consider the case that the weight of each node
is at most $(3\log n \cdot W)$. Thus, the weight of each node is an integer among the range $[0,\floor{3k\log
  n/\e}]$.  Note that by Lemma \ref{lm:headlog}, each node with weight at least $(3\log n \cdot W)$ must appear
in the optimal solution $\Omega^*$. We can directly append
such nodes and the nodes on the path from such nodes to root to our solution. Suppose the number of
these nodes are $k'$. The problem is reduced to find the $(k-k')$ nodes with maximum head value among
the remaining nodes which can be solved by the same method.

\end{document}